\documentclass{llncs}

\usepackage{amssymb}  
\usepackage[dvipdfmx]{graphicx}
\usepackage{algorithmic,algorithm} 
\usepackage{amsmath}
\usepackage{comment}
\usepackage{multirow}
\usepackage{subfig}

\usepackage{float}

\begin{document}
\title{Plane Formation by Synchronous Mobile Robots without
Chirality} 
\author{Yusaku Tomita \and Yukiko Yamauchi \and Shuji Kijima \and 
Masafumi Yamashita}
\institute{Graduate School of Information Science and Electrical
Engineering,\\  Kyushu University, Japan \\ 
\email{tomita@tcslab.csce.kyushu-u.ac.jp, \\
\{yamauchi, kijima, mak\}@inf.kyushu-u.ac.jp}
}

\maketitle

\begin{abstract}
 We consider a distributed system consisting
 of autonomous mobile computing entities,
 called robots, moving in a specified space. 
 The robots are {\em anonymous}, {\em oblivious}, and have neither any access
 to the global coordinate system nor any explicit communication medium.
 Each robot observes the positions of other robots
 and moves in terms of its {\em local coordinate system}.
 To investigate the self-organization power of robot systems,
 {\em formation problems} in the two dimensional space (2D-space) 
 have been extensively studied. 
 Yamauchi et al. (DISC 2015) introduced robot systems in the
 three dimensional space (3D-space).  
 While existing results for 3D-space assume 
 that the robots agree on the handedness of their local coordinate
 systems, 
 we remove the assumption and consider the robots without
 {\em chirality}. 
 One of the most fundamental agreement problems in 3D-space is
 the {\em plane formation problem} that requires the robots to land on a
 common plane, that is not predefined. 
 It has been shown that the solvability of the plane formation
 problem by robots with chirality
 is determined by the {\em rotation symmetry} of their 
 initial local coordinate systems because the robots cannot break it. 
 We show that when the robots lack chirality,
 the combination of rotation symmetry and {\em reflection symmetry} 
 determines the
 solvability of the plane formation problem 
 because a set of symmetric local coordinate systems without chirality
 is obtained by rotations and reflections. 
 This richer symmetry results in the increase of unsolvable instances
 compared with robots with chirality and
 a flaw of existing plane formation algorithm. 
 In this paper, we give a characterization of initial configurations
 from which the robots without chirality can form a plane and 
 a new plane formation algorithm for solvable instances. 
\end{abstract}

\section{Introduction} 

Distributed coordination of mobile computing entities
has been gaining increasing attention from
many areas such as robotics, transportation, construction, 
material engineering, DNA computing, and so on.
Though these wide areas of applications require complicated
operations, they can be classified into fundamental tasks, 
for example, gathering, formation, exploration, surveillance,
flocking, and partitioning. 
The underlying goals of these distributed coordination tasks are 
{\em agreement} and {\em self-organization}. 
We focus on a theoretical aspect of
one of such mobile computing entity models, 
called {\em autonomous mobile
robots}~\cite{CFPS12,DFSY15,FPSV14,FPSW08,FYOKY15,MV16,SY99,UYKY16,YS10,YUKY16,YUY16}.
A mobile robot system consists of a set of robots
each of which autonomously moves in a specified space.
Each robot is an {\em anonymous} (indistinguishable) point, 
and it executes a common distributed algorithm. 
Each robot repeats a {\em Look-Compute-Move cycle},
where it takes a snapshot of the positions of other robots in a Look phase,
computes its next position in the Compute phase, and
moves to the next position in the Move phase.
A {\em configuration} of such a system 
is the set of positions of the robots observed in
the global coordinate system, in other words,
a set of points. 
The robots have neither any access to the global
coordinate system nor any explicit communication medium.
Each robot observes and moves in terms of its
{\em local coordinate system}. 
Though observation is the only way for the robots to
cooperate with each other,
they have to tolerate inconsistency among
observations. 
A robot is {\em oblivious} if in a Compute phase,
it does not remember
the past observations and the past computations,
and can use the observation obtained in the Look phase
of the current cycle. 
Otherwise, a robot is {\em non-oblivious},
which means it is equipped with local memory. 
Existing literature introduces the following three
asynchrony models: 
In the {\em fully-synchronous (FSYNC) model},
the robots execute the $i$th Look-Compute-Move cycle
at the same time.
Thus the robots execute a cycle at each time step
$t=0, 1, 2, \ldots$. 
In the {\em semi-synchronous (SSYNC) model},
the robots follow discrete time steps,
but some robots may skip cycles. 
In the {\em asynchronous (ASYNC) model},
no assumption is made except that the length
of each cycle is finite.

The self-organization power of mobile robot systems has 
been studied for
robots in a discrete space (e.g., graphs)~\cite{DPT13,FIPS13},
in the two-dimensional space
(2D-space or plane)~\cite{CFPS12,DFSY15,FPSV14,FPSW08,FYOKY15,MV16,SY99,YS10,FYOKY15,SY99,YS10},
and in the three-dimensional space
(3D-space)~\cite{UYKY16,YUKY16,YUY16}. 
The {\em formation problem} requires the robots to
form a specified pattern from a given initial configuration.
The set of formable patterns indicates the 
self-organization power of a robot system.
Depending on the specified pattern, the formation
problem is classified into the following problems; 
the {\em point formation problem}, which
is the simplest form of the agreement problem
among the robots~\cite{CFPS12,FPSW05},
the {\em circle formation problem}~\cite{FPSV14,MV16}, 
and the {\em pattern formation problem} for arbitrary target
pattern~\cite{FPSW08,FYOKY15,SY99,YS10}.  
Since real systems work in 3D-space and 
applications such as drones become widely available, 
robot systems in 3D-space form an important and promising field. 
Yamauchi et al. proposed the {\em plane formation problem} 
that requires the robots to land on a common plane
without making any multiplicity.\footnote{As the
plane formation problem does not allow multiplicity,
point formation is not a solution. } 
The plane formation problem is one of the simplest agreement problems
in 3D-space and it bridges between
the robots in 3D-space and the robots in 2D-space,
so that existing techniques in 2D-space
can be used in 3D-space. 

In this paper, we consider the plane formation problem
by mobile robots that lack {\em chirality}.
A robot system does not have chirality when
the robots may not agree on the handedness (right-handed or left-handed)
of their local coordinate systems. 
On the other hand, a robot system has chirality when 
the handedness of all local coordinate systems are identical. 
Lack of chirality introduces heterogeneity among the robots, and 
the model is expected to 
reveal the self-organization power of the weakest robot model.
For example, Flocchini et al. and Mamino et al. showed
that more than four oblivious ASYNC robots 
can form a circle without chirality~\cite{FPSV14,MV16}. 

\begin{figure}[t]
 \centering
 \includegraphics[height=3cm]{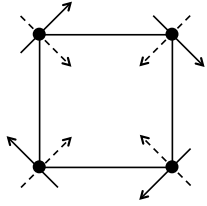}
 \caption{Symmetric initial positions and local coordinate systems.} 
\label{fig:local-square}
\end{figure}

Existing studies show that 
the set of formable patterns in 2D-space is
determined by the initial symmetry among the robots.
Consider an initial configuration of the four robots in 2D-space,
where they form a square and their local coordinate systems are
symmetric regarding the center of the square (Fig.~\ref{fig:local-square}). 
Since the robots execute a common algorithm,
from this initial configuration,
they keep square positions forever
if they execute cycles synchronously.
Yamashita et al. introduced the notion of 
{\em symmetricity} that gives formal explanation for such
situation~\cite{SY99,YS10}. 
We consider the decomposition of a set of points $P$
into regular $m$-gons
centered at one point. 
We consider that one point is a regular $1$-gon with an arbitrary center
and two points form a regular $2$-gon with the center being the 
midpoint. 
Then the maximum value of such $m$ is the symmetricity $\rho(P)$ of $P$
in 2D-space.
When $\rho(P)$ is greater than one, 
the common center is the center of the smallest enclosing
circle of $P$, denoted by $c(P)$, and
$\rho(P)$ is generally the order of the cyclic group that acts on
$P$. 
However, when $c(P) \in P$, this definition allows $\rho(P)=1$,
which means the symmetry of $P$ can be broken.
This is achieved by the robot on $c(P)$ leaving its current position.
It has been shown that irrespective of obliviousness and
asynchrony, 
the robots with chirality in 2D-space can form a target pattern $F$ from
an initial configuration $P$ if and only if $\rho(P)$ divides
$\rho(F)$ except the case where
$F$ is a point with multiplicity two~\cite{FYOKY15,SY99,YS10}.
The exception is called the {\em rendezvous problem},
which is trivially solvable by FSYNC robots while
not solvable by SSYNC (thus ASYNC) robots. 

The notion of symmetricity is later extended to 
the robots in 3D-space~\cite{YUY16}. 
In 3D-space, a set of symmetric local coordinate systems with
chirality is obtained by rotations on the global coordinate systems,
and there are five types of rotation symmetry; 
the {\em cyclic groups}, the {\em dihedral groups},
the {\em tetrahedral group}, the {\em octahedral group},
and the {\em icosahedral group}.
Each rotation symmetry forms a group that 
can be recognized as the set of symmetric
rotation operations on a prism, a pyramid,
a regular tetrahedron, a regular octahedron, and
a regular icosahedron, respectively.
In other words, each rotation group is determined by the
arrangement of rotation axes and their foldings. 
A rotation axis is a $k$-fold axis if it admits rotations
by $2\pi i/k$ for $i=1,2, \ldots, k$. 
Yamauchi et al. introduced rotation group and symmetricity
for a set of points $P$ in 3D-space. 
The {\em rotation group} $\gamma(P)$ of a set of points $P$ 
is the rotation group that acts on $P$ and none of its
supergroup in the set of rotation groups acts on $P$.
The {\em symmetricity} $\varrho(P)$ of $P$ 
is the set of rotation groups $G$ such that
the group action of $G$ on $P$ divides $P$ into $|G|$-sets
where $|G|$ is the order of $G$.
In the same way as 2D-space, the definition of symmetricity
implies symmetry breaking by movement of the robots 
because when some robots are on the rotation axes of $\gamma(P)$,
the robots do not allow the specified decomposition regarding
the rotation axis. In other words,
$\rho(P)$ consists of the rotation groups formed by
``unoccupied'' rotation axes of $\gamma(P)$.
Actually, the robots on rotation axes can remove the rotation axes
by leaving their current positions. 
Yamauchi et al. showed that irrespective of obliviousness,
the FSYNC robots with chirality can form a
target pattern $F$ from an initial configuration $P$
if and only if $\rho(P)$ is a subset of $\rho(F)$~\cite{YUY16}.

However, all these results assume chirality among the robots. 
After Yamashita et al. present
pattern formation algorithms for the oblivious SSYNC robots
with chirality in 2D-space~\cite{SY99,YS10},
Fujinaga et al. investigate
the {\em embedded pattern formation problem}, where
a target pattern is given as a set of landmarks 
on the plane~\cite{FOKY10}. 
They showed that oblivious ASYNC robots can form any embedded
target pattern by presenting an algorithm that is based on 
the ``clockwise'' 
minimum-weight perfect matching between the robots and the landmarks.
Based on this clockwise matching algorithm, 
Fujinaga et al. presented a pattern formation algorithm
for oblivious ASYNC robots with chirality~\cite{FYOKY15}. 
Later Cicerone et al. pointed out that the clockwise
matching algorithm does not work when the robots lack chirality, and 
showed a new embedded target pattern formation algorithm~\cite{CDN16}.
They also pointed out that robots without chirality 
may forever move symmetrically regarding an axis of symmetry. 

\noindent{\bf Our contribution.~}
The goal of our study is to formalize the 
degree of symmetry among the robots without chirality in 3D-space
and investigate their formation power. 
The contribution of this paper is twofold. 
First, we give a definition of symmetricity among the robots
without chirality in 3D-space. 
We consider both rotation symmetry and reflection symmetry 
because when the robots lack chirality, 
a local coordinate system is obtained by
a uniform scaling, a translation, a rotation,
a reflection by a mirror plane, 
or a combination of them
on the global coordinate system.\footnote{
When the robots have chirality, reflection is not necessary since
reflection changes the handedness of local coordinate system.} 
The combination of rotation symmetry and reflection symmetry 
introduces seventeen types of {\em symmetry groups}, which is
well studied in group theory and crystal symmetry~\cite{C97}. 
We extend the notion of symmetricity in \cite{YUY16}
to these seventeen symmetry groups. 
We validate the definition by showing
that the robots cannot resolve their symmetricity forever. 
Then, we give a necessary and sufficient condition
for FSYNC robots without chirality to solve the
plane formation problem.
To show the sufficiency, we present a new plane formation
algorithm since the existing plane formation 
algorithms for robots with chirality~\cite{UYKY16,YUKY16}
do not work in our model. 

We focus on the FSYNC robots with {\em rigid movement}, that is,
all robots synchronously execute a cycle and
reach their next positions in each cycle. 
If a robot stops en route, its movement is {\em non-rigid}.
While most existing results assume non-rigid movement, 
the worst case is when the robots cannot resolve their symmetry. 
Thus the worst case is determined by synchrony and rigid movement.
Formally, any execution of the FSYNC robots with rigid movement 
appears in the SSYNC (thus ASYNC) model with non-rigid movement. 

In \cite{YUKY16},
the cyclic groups and the dihedral groups are called
{\em 2D rotation groups} because 
one rotation axis is recognized,  
and when the rotation group of the current configuration
is a 2D rotation group, 
the robots with chirality can easily land on a ``horizontal'' plane
perpendicular to this rotation axis 
(Fig.~\ref{fig:land-antiprism} and Fig.~\ref{fig:land-prism-1}). 
On the other hand, the remaining three rotation groups
do not act on a set of points on a plane, and 
they are called {\em 3D rotation groups}. 
The necessary and sufficient condition in \cite{YUKY16}
is rephrased as follows: 
The FSYNC robots with chirality can form a plane from an initial
configuration $P$ 
if and only if $\varrho(P)$ consists of 2D rotation groups.
This characterization implies that
the FSYNC robots with chirality can form a plane from
an initial configuration where they form a regular polyhedron 
(except a regular icosahedron) or an icosidodecahedron,
while they cannot form a plane from the remaining (convex)
uniform polyhedra. 
Clearly,
the necessity of this result holds for the robots without chirality. 

\begin{figure}[t]
 \centering
\subfloat[]{\includegraphics[height=2.5cm]{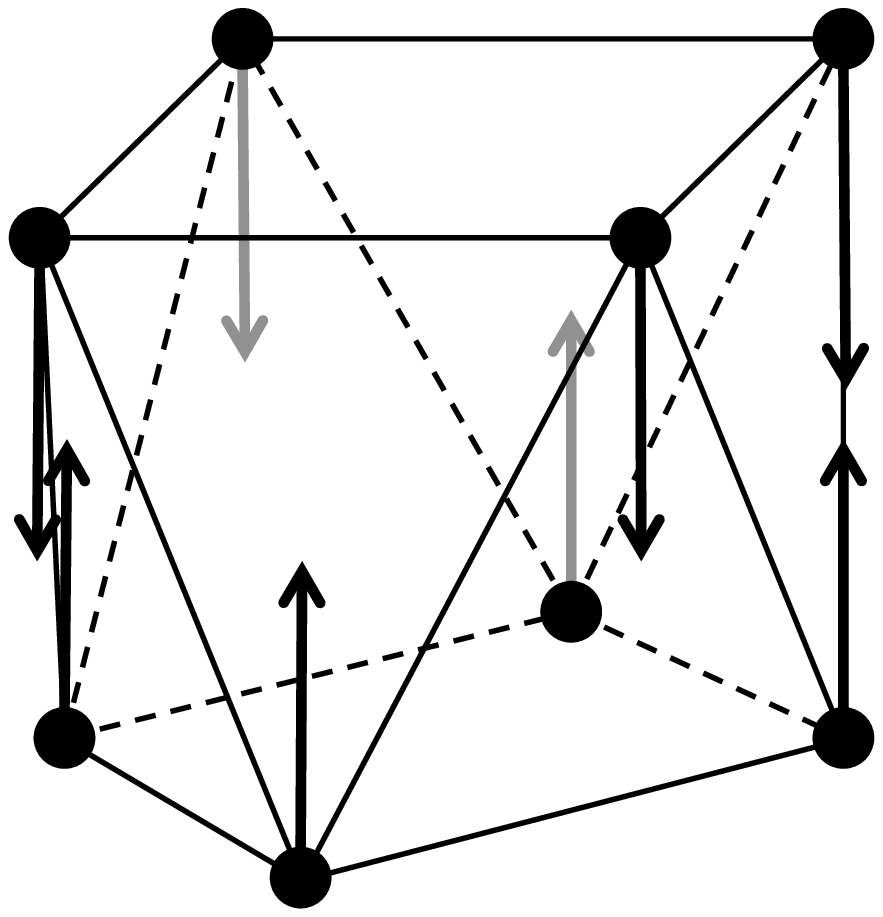}\label{fig:land-antiprism}}
 \quad 
\subfloat[]{\includegraphics[height=2.5cm]{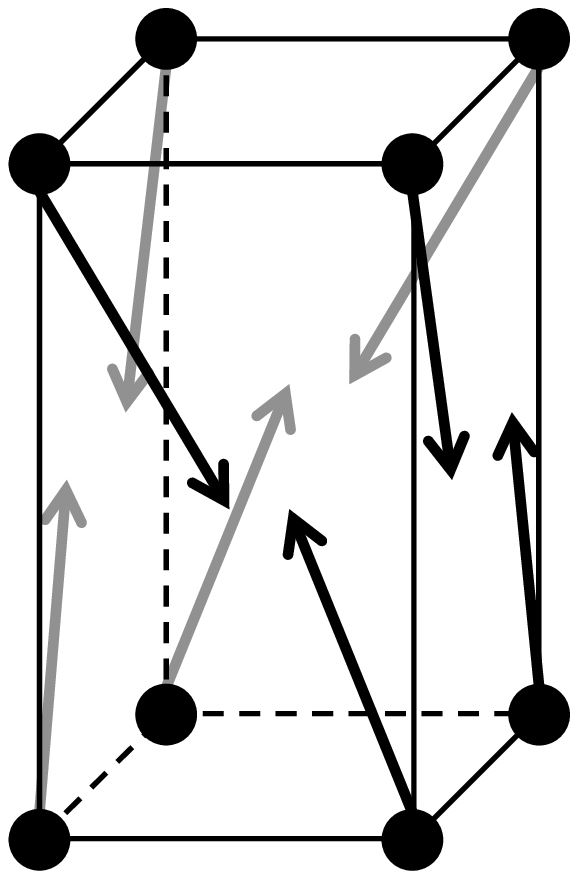}\label{fig:land-prism-1}}
 \quad 
\subfloat[]{\includegraphics[height=2.5cm]{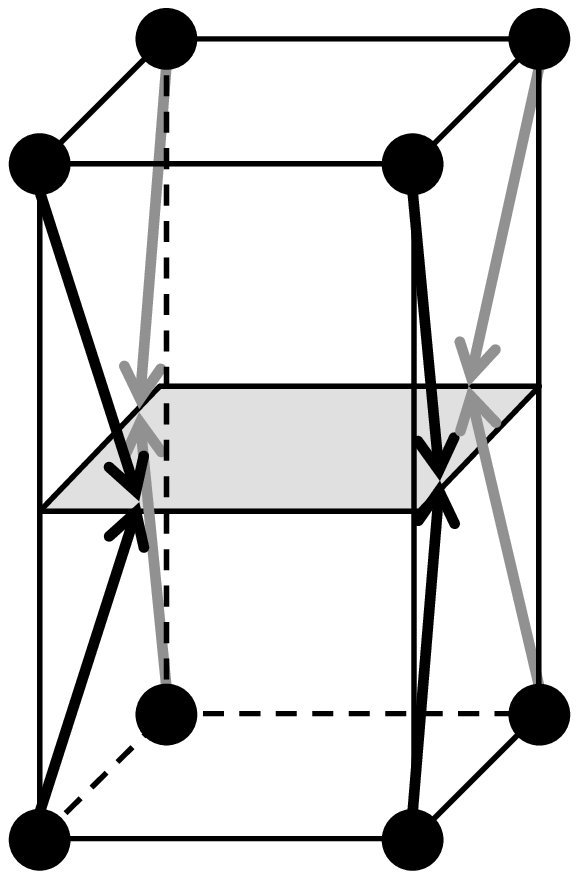}\label{fig:land-prism-2}}
 \quad 
\subfloat[]{\includegraphics[height=2.5cm]{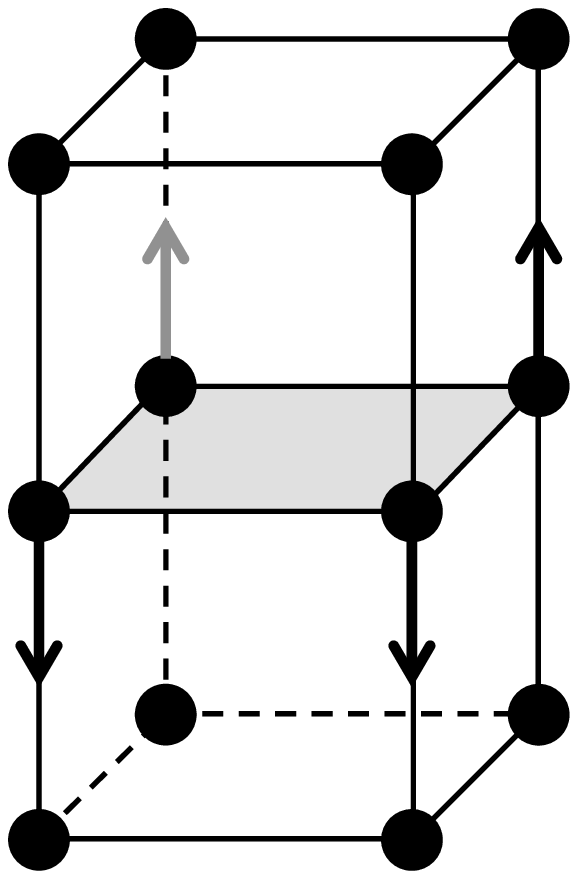}\label{fig:land-prism-3}}
 \caption{Examples of plane formation procedure. 
 (a) Robots directly land on a plane. 
 (b) Robots with chirality follow a ``right-screw rule''. 
 (c) Robots without chirality cannot avoid multiplicity.
 (d) Robots on the horizontal mirror plane can break
 the reflection symmetry. 
}
 \label{fig:examples-init}
\end{figure}

When the robots lack chirality, even when $\gamma(P)$ is a 2D rotation
group,
the ``horizontal'' plane can be a mirror plane and 
the robots cannot resolve the symmetry regarding 
this mirror plane (Fig.~\ref{fig:land-prism-2}). 
Actually, the only plane that the robots can agree is this
mirror plane, but they cannot avoid multiplicity on it.
As a result, a cube is removed from the set of solvable instances,
when compared with the robots with chirality. 
However, we will show that
when there is at least one robot on the horizontal mirror plane,
the robots can remove the mirror plane and can form a plane.
Intuitively, our necessary and sufficient condition
for the FSYNC robots without chirality
requires that an initial symmetricity contains
neither any 3D rotation group 
nor any combination of a 2D rotation group and 
an ``empty'' horizontal mirror plane. 
Our current results are preliminary in the sense of 
symmetry by rotation axes and mirror planes in 3D-space. 
For example, we defined the symmetricity among the robots,
and prove that the robots cannot resolve it,
but the symmetry breaking is not fully explored 
as we will address in the conclusion section. 

\noindent{\bf Organization.~}
In Section~\ref{sec:prel}, we define our robot model and
introduce rotation symmetry and reflection symmetry in 3D-space.
We present a necessary and sufficient condition for
plane formation by FSYNC robots without chirality.  
We show the necessity of the condition in Section~\ref{sec:nec},
and we prove the sufficiency by presenting a
plane formation algorithm in Section~\ref{sec:suf}. 
We conclude this paper with Section~\ref{sec:concl}.

\section{Preliminary}
\label{sec:prel}

\subsection{Robot Model} 

Let $R = \{r_1, r_2, \ldots, r_n\}$ be a set of $n$ anonymous 
robots, each of which is a point in 3D-space. 
We use $r_i$ just for description. 
We consider discrete time $t = 0, 1, 2, \ldots$ and
let $p_i(t) = (x_i(t), y_i(t), z_i(t)) \in \mathbb{R}^3$
be the position of $r_i$ at time $t$ in the global $x$-$y$-$z$
coordinate system
$Z_0$, 
where $\mathbb{R}$ is the set of real numbers.
The {\em configuration} of $R$ at time $t$ is
$P(t) = \{p_1(t), p_2(t), \ldots, p_n(t)\}$.
We denote the set of all possible configurations of $R$
by ${\mathcal P}^n$.
We assume that the initial positions of robots are distinct,
i.e., $p_i(0) \neq p_j(0)$ for $r_i \neq r_j$ and
$|P(0)| = n$.\footnote{When more than one robots are at one point,
it is impossible to separate them by a deterministic algorithm.} 
We also assume that $n \geq 4$ since any three robots are on one plane.

Each robot $r_i$ has no access to the global coordinate system, 
and it uses its local $x$-$y$-$z$ coordinate system $Z_i$.
The origin of $Z_i$ is the current position of $r_i$ while 
the unit distance, the directions, and the orientations of the
$x$, $y$, and $z$ axes of $Z_i$ are arbitrary and never change.
Hence, it is appropriate to denote $Z_i(t)$, but we use a shorter
description. 
Each $Z_i$ is either right-handed or left-handed. 
Thus the robots do not have {\em chirality}. 
We denote the coordinates of a point $p$ in $Z_i$ by $Z_i(p)$.

We consider the {\em fully-synchronous (FSYNC) model},
where the robots start the $t$th {\em Look-Compute-Move cycle}
at the beginning of time $(t-1)$ and finishes it
before time $t$ ($t = 1, 2, \ldots$). 
Each of the Look phase, 
the Compute phase, and the Move phase of a cycle
is completely synchronized at each time step. 
At time $t$, each robot $r_i$ obtains a set
$Z_i(P(t)) = \{Z_i(p_1(t)), Z_i(p_2(t)), \ldots, Z_i(p_n(t))\}$
in the Look phase.
Then $r_i$ computes its next position 
by using a common algorithm $\psi$ in the Compute phase.
A robot is {\em oblivious} if it does not remember
the past observations and the past computations, thus 
the input to $\psi$ is $Z_i(P(t))$.
Otherwise, it is {\em non-oblivious} and the input to $\psi$
contains the past observations and the past computations. 
Finally, $r_i$ moves to the next point in the Move phase.
We assume that each robot always reaches its next position in a
move phase and we do not care for the route to reach there. 
Thus we consider {\em rigid movement}. 

An execution of an algorithm $\psi$
from an initial configuration $P(0)$ is a sequence of configurations
$P(0), P(1), P(2), \ldots$.
When the initial local coordinate systems of $P(0)$, 
the algorithm $\psi$, and initial local memory content (if any)
are fixed, 
the FSYNC execution is uniquely determined. 

The {\em plane formation problem} requires that
the robots land on a plane, which is not predefined,
without making any multiplicity. Hence point formation is not
a solution for the plane formation problem. 
We say that an algorithm $\psi$ {\em forms a plane}
from an initial configuration $P(0)$,
if, regardless of the choice of initial local coordinate systems $Z_i$
for each $r_i \in R$, 
any execution $P(0), P(1), \ldots$ there exists a finite
$t \geq 0$ such that 
(i) $P(t)$ is contained in a plane, 
(ii) $|P(t)| = n$, i.e., all robots occupy distinct positions, and  
(iii) once the system reaches $P(t)$, the robots do not move anymore.

For a set of points $P$, we denote the smallest enclosing ball (SEB)
of $P$ by $B(P)$ and its center by $b(P)$.
A point on the sphere of a ball is said to be {\em on} the ball,
and we assume that the {\em interior} or the {\em exterior} of a ball
does not include its sphere.
The {\em innermost empty ball} $I(P)$ is the ball whose center is $b(P)$,
that contains no point of $P$ in its interior and 
contains at least one point of $P$ on its sphere. 
When all points of $P$ are on $B(P)$,
we say $P$ is {\em spherical}.

\subsection{Symmetry by Rotations and Reflections}

We consider symmetry among the robots which is caused
by not only symmetric positions of the robots
but also symmetric local coordinate systems of them. 
Since any local coordinate system is obtained by
a uniform scaling, a translation, a rotation,
a reflection by a mirror plane, 
or a combination of them
on the global coordinate system, we focus on symmetry operations by
rotation axes and mirror planes. 

A $k$-fold axis admits rotations by
$2\pi i/k$ ($i = 1,2, \ldots, k$). 
These $k$ operations form the {\em cyclic group} $C_k$
of order $k$.
When there are more than one rotation axes,
they also form a group, and there are five kinds of
rotation groups in 3D-space,
each of which is determined by the types of rotation
axes and the arrangement of them~\cite{C97}. 
Clearly, these multiple rotation axes intersect at one point. 
The {\em dihedral group} $D_{\ell}$ consists of a single $\ell$-fold
axis called the {\em principal axis} 
and $\ell$ $2$-fold axes perpendicular to the principal axis,
and its order is $2\ell$.\footnote{The rotation group $D_2$
consists of three $2$-fold rotation axes, and it has been
shown that the principal axis of $D_2$ can be also
recognized~\cite{YUKY16}.
This is because we do not consider $D_2$ only, but
a set of points and a rotation (or symmetry) group that acts on the
points. }
We can recognize $D_{\ell}$ by the rotations on a
prism with regular $\ell$-gon bases.
We abuse the term ``principal axis'' for the single rotation
axis of a cyclic group. 

The remaining three rotation groups are
the {\em tetrahedral group}, the {\em octahedral group},
and the {\em icosahedral group}, and we can recognize them by
the rotations on the corresponding regular polyhedra.
The {\em tetrahedral group} $T$ consists of three $2$-fold axes and
four $3$-fold axes, and its order is $12$.
The {\em octahedral group} $O$ consists of six $2$-fold axes,
four $3$-fold axes, and three $4$-fold axes, and its order is $24$.
The {\em icosahedral group} $I$ consists of fifteen $2$-fold axes,
ten $3$-fold axes, and six $5$-fold axes, and its order is $60$. 
We call the cyclic groups and the dihedral groups
{\em 2D rotation groups}, and
we call the remaining three rotation groups $T$, $O$, and $I$ 
{\em 3D rotation groups} 
because a 3D rotation group does not act on a point on a plane. 

A mirror plane changes the handedness and 
a mirror image of an object has a
different handedness from the original object. 
This is the reason why we need to consider reflection symmetry
when we consider the robots without chirality. 
The {\em bilateral symmetry} $C_s$ consists of one mirror plane
and its order is $2$. 
When there are more than one mirror planes, 
an intersection of mirror planes introduces a rotation axis.
We consider the compositions of rotation symmetry and
mirror planes.
Each symmetry type also forms a group. 
Clearly, the rotation axes and mirror planes of the symmetry type 
intersect at one point. 
The composition of $C_k$ ($k>1$) and a mirror plane perpendicular
to the principal axis is denoted by $C_{kh}$, where 
$h$ represents the ``horizontal'' mirror plane.
The order of $C_{kh}$ is $2k$. 
The composition of $C_k$ ($k>1$) and $k$ mirror planes 
containing the principal axis is denoted by $C_{kv}$, where
$v$ represents the ``vertical'' mirror planes.
The order of $C_{kv}$ is $2k$. 
The composition of $D_{\ell}$ ($\ell \geq 2$) and
a horizontal mirror plane regarding the principal axis
is denoted by $D_{\ell h}$.
However, this horizontal mirror plane together with
rotation axes of $D_{\ell}$ forces $\ell$ vertical mirror planes
each of which contains two $2$-fold axes and the principal axis. 
The order of $D_{\ell h}$ is $4 \ell$. 
The composition of $D_{\ell}$ ($\ell \geq 2$) and
$\ell$ vertical mirror planes is denoted by $D_{\ell v}$.
The vertical mirror planes do not contain any $2$-fold axes,
otherwise the rotation axes of $D_{\ell}$ forces a horizontal mirror
plane. 
The order of $D_{\ell v}$ is $4 \ell$. 

The composition of $T$ and three mutually perpendicular
mirror planes, each of which contains two $2$-fold 
axes is denoted by $T_h$.
The order of $T_h$ is $24$. 
The composition of $T$ and six mirror planes, each of which
contains two $3$-fold axes
is denoted by $T_d$.
The order of $T_d$ is $24$. 
The composition of $O$ and nine mirror planes is denoted by
$O_h$.
Three of the mirror planes are mutually perpendicular
and each of them contains two $4$-fold axes.
Each of the remaining six mirror planes contains
two $3$-fold axes. 
The order of $O_h$ is $48$. 
The composition of $I$ and fifteen mirror planes, each of which
contains two $5$-fold axes is denoted by $I_h$. 
The order of $I_h$ is $120$.

Another type of composite symmetry is {\em rotation reflection},
where a rotation regarding a single rotation axis
and taking a mirror image regarding a mirror plane perpendicular to
the rotation axis are alternated.
This type of symmetry is denoted by $S_{m}$.
Because of the alternation, the folding of the rotation axis is even. 
$S_2$ corresponds to the central inversion, which is denoted by
$C_i$. See Appendix~\ref{sec:list-of-symmetries} for
more detail. 

Let ${\mathbb S} = \{C_1, C_i, C_s, C_k, C_{k h}, C_{k v},
D_{\ell}, D_{\ell h}, D_{\ell v}, S_m, T, T_d, T_h, O, O_h, I, I_h ~|~$ 
$k=2, 3, \ldots, \ell = 2, 3, \ldots\}$
where $C_1$ consists of only the identity element. 
We call the elements of ${\mathbb S}$ {\em symmetry groups}. 
These seventeen types of symmetry groups
describe all symmetry in 3D-space~\cite{C97}. 
In this paper we consider
rotation symmetry separately. 
We call the elements of
$\{C_k, D_{\ell}, T, O, I ~|~ k=1, 2, \ldots, \ell=2, 3, \ldots\}$
the {\em rotation groups}. 

We denote the order of $G \in {\mathbb S}$ with $|G|$. 
When $G'$ is a subgroup of $G$
($G, G' \in {\mathbb S}$), we denote it 
by $G' \preceq G$.
If $G'$ is a proper subgroup of $G$ (i.e., $G \neq G'$), 
we denote it by $G' \prec G$.
For example, we have $D_2 \prec T$, 
$T \prec O, I$, but $O \not\preceq I$.
If $G \in {\mathbb S}$ has a $k$-fold axis, 
$C_{k'} \preceq G$ if $k'$ divides $k$. 
For symmetry groups containing mirror planes,
$T \prec T_h \prec O_h$ but $T_h \not\prec O$. 
For $S_m$, we have $C_{m/2} \prec S_m \prec C_{mh}$.

\subsection{Rotation Group, Symmetry Group, and Symmetricity}

Let $P \in {\mathcal P}^n$ be a set of $n$ points.
The {\em rotation group} $\gamma(P)$ of $P$ is the rotation group 
that acts on $P$ and none of its proper supergroup in
$\{C_k, D_{\ell}, T, O, I ~|~ k=1, 2, \ldots, \ell=2, 3, \ldots\}$ 
acts on $P$. 
The {\em symmetry group} $\theta(P)$ of $P$ is the symmetry group
that acts on $P$ and none of its proper supergroup
in ${\mathbb S}$ acts on $P$. 
Clearly, $\gamma(P)$ is a subgroup of $\theta(P)$ 
($\gamma(P) \preceq \theta(P)$), and 
they are uniquely determined.\footnote{See
for example~\cite{C97}, that shows an algorithm to uniquely determine
the symmetry group of a polyhedra.
The algorithm checks 
rotation axes, mirror planes, and a point of inversion. 
Since we consider a set of points and their convex-hulls, 
we can use the same algorithm. } 
By the definition, when $\theta(P)$ is either $C_1, C_i, C_s$,
$\gamma(P) = C_1$ because such configuration $P$ does not have
any rotation axis. 
Table~\ref{table:sgs-regular} 
shows the rotation group of a set of vertices of each regular polyhedron.

 \begin{table}[t]
   \centering
 \caption{Rotation group, symmetry group, and symmetricity
 of regular polyhedra}
 \label{table:sgs-regular}
 \centering 
 \begin{tabular}{|l|c|c|c|}
  \hline 
  Polyhedron & Rotation group & Symmetry group & Symmetricity \\
  \hline 
  Regular tetrahedron & $T$ & $T_d$ & $\{D_2, S_4\}$\\
  \hline 
  Regular octahedron & $O$ & $O_h$ & $\{D_3, S_6\}$\\
  \hline 
  Cube & $O$ & $O_h$ & $\{D_4, D_{2h}, D_{2v}, C_{4h}, S_4\}$\\
  \hline 
  Regular dodecahedron & $I$ & $I_h$ & $\{D_5, D_2, S_{10}\}$\\
  \hline 
  Regular icosahedron & $I$ & $I_h$ & $\{T, D_3, S_6\}$\\ 
  \hline 
 \end{tabular}
 \end{table}

The group action of $\theta(P)$ 
decomposes $P$ into disjoint subsets. 
Let $Orb(p) = \{g*p ~|~ g \in \theta(P)\}$ be the orbit of $p \in P$
where $*$ denotes the action of $g$ on $s$, and
the orbit space $\{Orb(p) ~|~ p \in P\} = \{P_1, P_2, \ldots, P_m\}$
is called the {\em $\theta(P)$-decomposition} of $P$.
Each element $P_i$ is {\em transitive} because it is one orbit 
regarding $\theta(P)$. 

Yamauchi et al. showed that in configuration $P$ without any multiplicity, 
the robots with chirality can agree on the $\theta(P)$-decomposition
$\{P_1, P_2, \ldots, P_m\}$ of $P$ and
a total ordering among the elements so that 
(i) $P_1$ is on $I(P)$, (ii) $P_m$ is on $B(P)$,
and (iii) $P_{i+1}$ is not in the interior of the
ball centered at $b(P)$ and containing $P_i$ on its sphere~\cite{YUKY16}.
Though their technique relies on chirality, 
we can extend it to robots without chirality.
In~\cite{YUKY16}, each robot translates its local observations
to a ``celestial map'' by considering $I(P)$ as the earth
and its current position is on the half line from
$b(P)$ containing the north pole.
Then, the robot selects an appropriate robot to define the prime
meridian and translates the position of each robot to 
a triple consisting of its altitude, latitude, and longitude.
The ordered sequence of these triples is the {\em local observation} of
the robots. 
However, the lack of chirality does not allow the robots to
agree on the direction of longitude.
Then we make a robot consider both directions and 
select the direction that produces the smallest sequence. 
In the same way as \cite{YUKY16}, we have the
following property. 

\begin{lemma}
\label{lemma:ordering}  
Let $P \in {\mathcal P}_n^3$ and $\{ P_1, P_2, \ldots, P_m \}$ be 
a configuration of robots represented as a set of points 
and its $\theta(P)$-decomposition, respectively. 
Then we have the following two properties: 
\begin{enumerate}
\item 
For each $P_i$ ($i = 1, 2, \ldots , m$), 
all robots in $P_i$ have the same local view.
\item
Any two robots, one in $P_i$ and the other in $P_j$, 
have different local views, for all $i \not= j$.
\end{enumerate}
\end{lemma}
\begin{proof}
The first property is obvious by the definitions of 
$\theta(P)$-decomposition and local view,
since for any $p, q \in P_i$ there is an element
$g \in \theta(P)$ such that $q = g*p$.

As for the second property,
to derive a contradiction,
suppose that there are distinct integers $i$ and $j$,
such that robots $r_k \in P_i$ and  $r_{\ell} \in P_j$ have the same 
local view. 
Let $V_k^*$ and $V_{\ell}^*$ be the local view of
$r_k$ and $r_{\ell}$. Thus, we have $V_k^* = V_{\ell}^*$.
Let us consider a function $f$ that maps the $d$th element of
$V_k^*$ to that of $V_{\ell}^*$.
More formally, letting the $d$th element of $V_k^*$ (resp., $V_{\ell}^*$)
be $p_x^*$ (resp., $p_y^*$), 
$f$ maps $p_x$ to $p_y$.
Then $f$ is a transformation that keeps $b(P)$ unchanged
by the definition of local view,
i.e., $f$ is a rotation or an reflection in $\theta(P)$,
which contradicts to the definition of $\theta(P)$-decomposition.
\qed 
\end{proof}

By Lemma~\ref{lemma:ordering}, the robots can agree on a total ordering
of the elements of the $\theta(P)$-decomposition of $P$. 
In the following, we assume that $\{P_1, P_2, \ldots, P_m\}$
is sorted by this ordering. 

We denote the set of local coordinate systems for
configuration $P$ with
$Q = \{(o_i, x_i, y_i, z_i ~|~ p_i \in P)\}$
where $o_i$ is the position of $p_i \in P$
(i.e., the origin of $Z_i$) and 
$x_i$, $y_i$, and $z_i$ are the
$(1, 0 , 0)$, $(0, 1, 0)$, and $(0, 0, 1)$ of $Z_i$
observed in the global coordinate system $Z_0$.
We use $(P, Q)$ to explicitly show the set of local coordinate
systems for $P$ though $Q$ contains $P$ as
$\{o_1, o_2, \ldots, o_n\}$.
We define the {\em symmetry group} $\sigma(P, Q)$
of $(P, Q)$ as the symmetry group of that acts on $(P, Q)$ and
none of its proper supergroup in ${\mathbb S}$ acts on it.
Clearly, we have $\sigma(P, Q) \preceq \theta(P, Q)$. 
We define the $\sigma(P,Q)$-decomposition of $(P,Q)$
in the same way as the $\theta(P)$-decomposition of $P$.
We note that the robots of $P$ cannot obtain $Q$ nor $\sigma(P,Q)$
because they can observe only the positions of themselves.

Given a set $P$ of points, 
$\theta(P)$ determines the arrangement of its rotation axes
and mirror planes in $P$. 
We thus use $\theta(P)$ and the arrangement 
of its rotation axes and mirror planes in $P$ interchangeably.
For two groups $G, H \in {\mathbb S}$, 
an {\em embedding} of $G$ to $H$
is an embedding of each rotation axis and each mirror plane 
of $G$ to one of the rotation axes and one of the mirror planes of $H$ 
with keeping their arrangement in $G$.
Any $k$-fold axis of $G$ is embedded so that it
overlaps a $k'$-fold axis of $H$, where $k$ divides $k'$, and 
any mirror plane of $G$ is embedded to a mirror plane of $H$.
However we need careful treatment for $S_m$. 
When $G = S_m$, its mirror plane is embedded to a mirror plane of $H$,
and when $H = S_m$, its mirror plane cannot grant any mirror plane of $G$. 
For example, we can embed $T$ to $O$. 
There are three embeddings of $C_4$ to $O$ depending on the choice of 
the $4$-fold axis. 
We can embed $C_3$ to $S_6$, and $S_6$ to $C_{6h}$.
We can embed $D_{\ell v}$ to $D_{2\ell h}$ but cannot
to $D_{\ell h}$. 
Observe that we can embed $G$ to $H$ if and only if $G \preceq H$. 

We can also consider a $G$-decomposition of a set of points $P$ 
for some $G \prec \theta(P)$ for an embedding of $G$ in
$\theta(P)$. 
However, for such $G$-decomposition,
the robots cannot agree the ordering among the elements
since Lemma~\ref{lemma:ordering} does not hold. 

We now define {\em symmetricity} of a set of points in 3D-space.
Intuitively, symmetricity shows all possible symmetry groups
to which the robots may forever subject.
As the symmetry groups are partially ordered,
we consider a set of such rotation groups. 

\begin{definition}
\label{def:symmetricity}
Let $P \in {\cal P}_n^3$ be a set of points. 
The symmetricity $\varrho(P)$ of $P$
is the set of symmetry groups $G \in {\mathbb S}$ 
that acts on $P$ (thus $G \preceq \theta(P)$)
and there exists an embedding of $G$ to
$\theta(P)$ such that each element of the $G$-decomposition of $P$ 
is a $|G|$-set. 
\end{definition}

We define $\varrho(P)$ as a set because the ``maximal'' symmetry 
group that satisfies the definition is not uniquely determined. 
Maximality means that there is no proper supergroup in
${\mathbb S}$ that satisfies the condition of
Definition~\ref{def:symmetricity}. 
When it is clear from the context, we denote $\varrho(P)$ by
the set of such maximal elements. 
For example, if $P$ forms a cube,
\begin{equation*}
\varrho(P) =
\{C_1, C_i, C_2, C_4, C_{2h}, C_{4h}, C_{2v},
D_2, D_4, D_{2h}, D_{2v},
S_4\}, 
\end{equation*}
and we denote it by $\varrho(P) = \{D_4, D_{2h}, D_{2v}, C_{4h}, S_4\}$. 
The set $\varrho(P)$ does not contain $O$ itself
since $O$-decomposition of $P$ consists of one 
$8$-set, while $|O| = 24$. 
From the definition, $\varrho(P)$ contains every element of
${\mathbb S}$ that is a subgroup of $G$ if $G \in \varrho(P)$. 
See Table~\ref{table:sgs-regular} as an example. 

We can rephrase the definition of symmetricity of a set of points $P$
as a set of symmetry groups formed by rotation axes and
mirror planes of $\theta(P)$ that do not contain any
point of $P$.
This is because a point on a rotation axis
(a mirror plane, respectively) 
does not allow a decomposition into $|G|$-sets
for any $G$ containing the rotation axis. \footnote{
We assume that a set of points $P$ does not contain
any multiplicity. In other words,
we consider an initial configuration $P$. }

We conclude this section with the following two lemmas,
that validate the definition of symmetricity. 
Lemma~\ref{lemma:rho-sigma} shows that there exists an
arrangement of local coordinate systems $Q$ for any
initial configuration $P$ and $G \in \varrho(P)$
such that $\sigma(P, Q) = G$.
Then, Lemma~\ref{lemma:sigma-forever} shows that
the robots are caught in this initial symmetry. 

In the proofs, we take a new view of the positions of the robots.
In the definition of symmetricity $\varrho(P)$
for a set of points $P$, 
we consider an arrangement of a symmetry group in $P$.
On the other hand, to show that the initial symmetry cannot be 
broken, we consider the cases where the positions of robots are 
caught in an arrangement of $G \in \varrho(P)$. 

\begin{lemma}
 \label{lemma:rho-sigma}
 For an arbitrary initial configuration $P \in {\mathcal P}^n$
 and any $G \in \varrho(P)$,
 there exists a set of local coordinate systems $Q$ 
 such that $\sigma(P, Q) = G$. 
\end{lemma} 
\begin{proof}
 We show a construction of $Q$ for $P$ and $G$. 
 Let $\{P_1, P_2, \ldots, P_m\}$ be the $G$-decomposition of $P$
 for some embedding of $G$ to $\theta(P)$.
 Clearly, such embedding exists since $G \preceq \theta(P)$.
 From the definition, $|P_j| = |G|$ for $j=1, 2, \ldots, m$.
 For each $P_j$, we arbitrary fix a local coordinate system
 of one robot $p_i \in P_j$.
 Then for each $p_k \in P_j$, there exists a unique
 element of $G$ such that $p_k = g_k * p_i$ and
 $g_k \neq g_{\ell}$ if $p_k \neq p_{\ell}$ for any $p_{\ell} \in P_j$.
 Then we fix the local coordinate system of $p_k$
 by applying $g_k$ to the local coordinate system of $p_i$.
 The local coordinate systems $Q$ obtained by this procedure
 satisfies the property.
 \qed 
\end{proof}

\begin{lemma}
 \label{lemma:sigma-forever}
 Irrespective of obliviousness, 
 for an arbitrary initial configuration $P \in {\mathcal P}^n$, 
 any $G \in \varrho(P)$, and any algorithm $\psi$, 
 there exists an execution $P(0)(=P), P(1), P(2), \ldots$
 such that $\theta(P(t)) \succeq G$. 
\end{lemma}
\begin{proof}
 Let $Q$ be initial local coordinate systems for $P$ such that
 $\sigma(Q, P) = G$ for arbitrary $G \in \varrho(P)$.
 By Lemma~\ref{lemma:rho-sigma},
 such $Q$ always exists.
 Let $\{P_1, P_2, \ldots, P_m\}$ be the $G$-decomposition of
 $P$. 
 
 From this arrangement of initial local coordinate systems, 
 the robots forming $P_j$ keeps their symmetry group $G$
 forever for any algorithm $\psi$.
 We first show an induction for the oblivious FSYNC robots.
 For any $p_i, p_k \in P_j$,
 when $\psi(Z_i[P(0)]) = x$ holds, we have
 $\psi(Z_k[P(0)]) = x$ and 
   $Z_0[Z_k(\psi(Z_k[P(0)]))] = g_k * Z_0[Z_i(\psi(Z_i[P(0)]))]$.
 Let $P_j(1) \subseteq P(1)$ be the positions of robots of $P_j$
 in $P(1)$. 
 Thus $\theta(P_i(1)) = G$ and $\theta(P(1)) \succeq G$.
 By an easy induction for $t = 1, 2, \ldots$,
 we have the property for any $P(t)$. 
 
 Non-obliviousness does not improve the situation.
 When the initial memory contents are identical (for example, empty),
 the above discussion holds for the transition from $P(0)$ to $P(1)$.
 During this transition,
 the robots in the same element $P_j$ of the
 $G$-decomposition of $P$ obtain the same
 local observation, performs the same computation, and
 exhibits the same movement.
 Thus, their local memory content are still the same in $P(1)$,
 and they continue symmetric movement during the
 transition from $P(1)$ to $P(2)$.
 \qed 
\end{proof}

\section{Impossibility of Plane Formation}
\label{sec:nec}

The following theorem shows a necessary condition
for the FSYNC robots without chirality to form a plane,
that will be shown to be a sufficient condition in
Section~\ref{sec:suf}.
The condition means that to solve the
plane formation from an initial configuration $P$, 
$\varrho(P)$ should not contain any of the following
symmetry groups:
$T$, $T_d$, $T_h$, $O$, $O_h$, $I$, $I_h$, 
$C_{k h}$ ($k \geq 3$), and $D_{\ell h}$ ($\ell \geq 2$).

\begin{theorem}
 \label{theorem:nec}
 Irrespective of obliviousness,
 the FSYNC robots without chirality can form a plane
 from an initial configuration $P$ 
 only if $\rho(P)$ consists of
 $C_1$, $C_s$, $C_i$, $C_k$, $C_{k v}$, $C_{2h}$,
 $D_{\ell}$, $D_{\ell v}$, and $S_m$. 
\end{theorem}

\begin{proof}
 Let $\psi$ be an arbitrary plane formation algorithm for
 an initial configuration $P$ such that
 $\varrho(P)$ contains a 3D rotation group, 
 $C_{k h}$ ($k \geq 3$), or $D_{\ell h}$ ($\ell \geq 2$).
 We have the following three cases. 
 
 \noindent{\bf Case A:~} $\varrho(P)$ contains $C_{k h}$
 for some $k \geq 3$. 
   
 Let $Q$ be a set of initial local coordinate systems for $P$
 such that $\sigma(P, Q) = C_{k h} \in \varrho(P)$ ($k \geq 3$).
 From Lemma~\ref{lemma:sigma-forever},
 irrespective of obliviousness, 
 for any algorithm $\psi$, there exists an execution
 $P=P(0), P(1), P(2), \ldots$ such that
 $\theta(P(t)) \succeq C_{k h}$ for any $t \geq 0$. 
 Assume that $P(t')$ be a terminal configuration. 
 Then $\theta(P(t'))$ is a supergroup of $C_{k h}$, and 
 $\theta(P(t'))$ has the mirror plane of $\theta(P)$. 
 The robots are on the mirror plane,
 otherwise the robots are not on one plane
 because of their symmetry. 
 Let $\{P_1, P_2, \ldots, P_m\}$ be
 the $\sigma(P, Q)$-decomposition of $P(=P(0))$.
 For each $P_i$ ($1 \leq i \leq m$) and $p \in P_i$,
 there exists $q \in P_i$ such that $p$ and $q$ are at
 symmetric positions regarding the mirror plane of $C_{k h}$. 

 By Lemma~\ref{lemma:sigma-forever},
 the robots of $P_i$ move with keeping the rotation axis and the
 mirror plane of the embedding of $C_{k h}$ in $P$. 
 Thus $p$ and $q$ occupy the same point on the mirror plane of
 $C_{k h}$ in $P(t')$. 
 Hence the robots cannot avoid multiplicity and
 $P(t')$ is not a terminal configuration of the
 plane formation problem. 
 
 \noindent{\bf Case B:~} $\varrho(P)$ contains $D_{\ell h}$
 for some $h \geq 2$.
 
 Let $Q$ be a set of initial local coordinate systems for $P$
 such that $\sigma(P, Q) = D_{\ell h} \in \varrho(P)$ ($\ell \geq 2$).
 By Lemma~\ref{lemma:sigma-forever}, irrespective of obliviousness, 
 for any algorithm $\psi$, there exists an execution
 $P=P(0), P(1), P(2), \ldots$ such that
 $\theta(P(t)) \succeq D_{\ell h}$ for any $t \geq 0$.
 We have the same discussion as Case A. 
 If there exists a terminal configuration, 
 the robots are on the initial horizontal mirror plane of $D_{\ell h}$.
 Hence, the robots cannot avoid multiplicity and
 $P(t')$ is not a terminal configuration of the
 plane formation problem. 
  
 \noindent{\bf Case C:~} $\varrho(P)$ contains a 3D-rotation
 group.
 
 The impossibility for this case has been shown for
 robots with chirality in~\cite{YUKY16}
 and the result holds for our robots because
 our model allows the robots with chirality. 
 We note that when $\varrho(P)$ contains 
 $T_d, T_h, O_h$ or $I_h$, then it contains
 the corresponding rotation group because it is a subgroup
 without any mirror plane.
 \qed 
\end{proof}

\section{Plane Formation Algorithm}
\label{sec:suf}

In this section, we show a plane formation algorithm for
oblivious FSYNC robots without chirality 
and prove our main theorem.

\begin{theorem}
 \label{theorem:main}
 Irrespective of obliviousness,
 the FSYNC robots without chirality can form a plane
 from an initial configuration $P$ 
 if and only if $\varrho(P)$ consists of
 $C_1$, $C_i$, $C_s$, $C_k$, $C_{k v}$, $C_{2h}$,
 $D_{\ell}$, $D_{\ell v}$, and $S_m$. 
\end{theorem}

The necessity is clear from Theorem~\ref{theorem:nec}.
We prove the sufficiency by presenting a plane formation algorithm
for solvable instances (i.e., initial configurations). 
Because of the condition of the theorem, 
solvable instances are classified into the following three types.

\begin{description}
 \item[Type 1:~] Initial configurations with 3D rotation groups. 
	    From the condition of Theorem~\ref{theorem:main},
	    any initial configuration $P$ of this type contains one of the
	    following polyhedra as an element of its 
	    $\theta(P)$-decomposition 
	    because some robots are on some rotation axes:
	    a regular tetrahedron,
	    a regular octahedron,
	    a regular dodecahedron, and an icosidodecahedron.\footnote{
	    Points on rotation axes of a 3D rotation group
	    also forms a cube, a cuboctahedron, and a regular icosahedron.
	    However, these polyhedra allow $T$($\prec O, I$)
	    to join their symmetricity. } 
 \item[Type 2:~] Initial configurations with 2D rotation groups
	    with at least one rotation axis.
	    From the condition of Theorem~\ref{theorem:main},
	    any initial configuration $P$ of this type satisfies that
	    either $\theta(P)$ 
	    does not have the horizontal mirror plane or 
	    there are some robots on a horizontal mirror plane. 
 \item[Type 3:~] Initial configurations without any rotation axis.
	    This case is further divided into asymmetric initial
	    configurations, 
	    a symmetric initial configurations with a single mirror plane,
	    and a symmetric initial configurations with point of
	    inversion. 
\end{description}

The proposed algorithm handles these three types separately.
The robots can agree on the type of the current configuration and
they execute the corresponding algorithm.
For the first case, the robots first break their symmetry and
translates an initial Type 1 configuration to another Type 2 or Type 3 
configuration. 
For the second case, the robots agree on a plane
perpendicular to the principal axis and containing the center of
their smallest enclosing ball, and land on the plane. 
For the third case, if the initial configuration is
asymmetric, the robots agree on a plane by using
the total ordering among themselves.
Otherwise, the robots agree on a plane other than
the mirror plane by using two elements of their
decomposition and lands on it.

Before we go into the detailed description of the proposed algorithm,
we show preparation steps for an initial configuration $P$. 
These steps can be realized very easily in the FSYNC model 
because the set of robots to move is easy to recognize,
and the movement neither makes collisions nor 
changes the symmetry group and the symmetricity of the robots. 
In the following, we use a point and a robot at the point
interchangeably.
For example, the position $p_i \in P$ of robot $r_i$ means
$r_i$, and $P' \subseteq P$ means the set of robots
at positions of $P'$. 

First, when $b(P) \in P$, the preparation phase sends the robot on $b(P)$
to an arbitrary point in the interior of $I(P)$,
so that a resulting configuration will be asymmetric.
The next position of the robot at $b(P)$ is,
for example, a point neither on any rotation axis nor 
on any mirror plane of $\theta(P \setminus \{b(P)\})$.

Second, for Type 1 cases, the preparation phase moves an element 
$P_i$ of the $\theta(P)$-decomposition of $P$ forming
one of the specified polyhedra to the interior of $I(P)$.

Finally, for Type 2 cases, the preparation phase moves an element 
$P_j$ of the $\theta(P)$-decomposition of $P$ on the horizontal mirror
plane of $\theta(P)$ to the interior of $I(P)$.
At the same time, if there exists another element $P_k$ of the
$\theta(P)$-decomposition of $P$ that is on neither 
the horizontal mirror plane nor the principal rotation axis,
the preparation phase makes the element slightly ``shrink'' 
so that $\theta(P)$ is kept by $P_k$.
For example, if $P$ with $\theta(P) = D_{4h}$
consists of a cube and a square on the horizontal mirror plane,
the robots forming a cube shrink to translate the cube 
to a long square prism. 

For the second and the third cases,
we select the minimum index among the elements satisfying the condition 
and move each $p \in P_i$ ($P_j$, respectively)
along the line $\overline{p b(P)}$. 
Since we select the minimum index,
this movement introduces no collision.
Additionally, in the third case, 
the robots of $P_j$ move on the horizontal mirror plane
toward $b(P)$.

\subsection{Symmetry Breaking}

We consider a Type 1 configuration $P$.
Let $\{P_1, P_2, \ldots, P_m\}$ the the $\theta(P)$-decomposition of
$P$. 
The preparation step guarantees that
$P_1$ forms one of the following four polyhedra;  
a regular tetrahedron, 
a regular octahedron, 
a regular dodecahedron, and 
an icosidodecahedron. 
Then the proposed plane formation algorithm first makes the
robots execute the {\em go-to-center algorithm} 
(Algorithm~\ref{alg:go-to-center})
proposed in \cite{YUKY16}. 
Each robot of $P_1$ selects an adjacent face of the
polyhedron and moves to the center of the selected face.
But it stops $\epsilon$ before the center to avoid collisions. 

Algorithm~\ref{alg:go-to-center} does not depend on chirality
among the robots and it has been shown that
the rotation group of a resulting
configuration is always a 2D rotation group~\cite{YUKY16}.
Since our robots lack chirality, 
we should consider the combination of such 2D rotation groups
and mirror planes. 
The following lemma guarantees that any resulting configuration
does not have any horizontal mirror plane (except $C_{2h}$).
We note that we do not have to care for rotation reflections
in this phase because such configurations are
Type 3 configurations.

\begin{algorithm}[t]
\caption{Symmetry breaking algorithm for robot $r_i$~\cite{YUKY16}}
\label{alg:go-to-center}
\begin{tabbing}
xxx \= xxx \= xxx \= xxx \= xxx \= xxx \= xxx \= xxx \= xxx \= xxx
\kill 
{\bf Notation} \\ 
 \> $P$: Current configuration observed in $Z_i$. \\
 \> $p_i \in P$: The position of $r_i$ (i.e., the origin of $Z_i$). \\
 \> $\{P_1, P_2, \ldots, P_m\}$: $\theta(P)$-decomposition of $P$, where
 $P_1$ forms one of the four \\
 \> \> \> \> \> polyhedra. \\
 \> $\epsilon = \ell /100$, where $\ell$ is the length of an edge
 of the polyhedron that $P_1$ forms.  \\ 
 {\bf Algorithm}  \\ 
\> {\bf If} $p_i \in P_1$ {\bf then} \\
\> \> {\bf If} $P_1$ is an icosidodecahedron {\bf then}\\ 
\> \> \> Select an adjacent regular pentagon face of $P_1$. \\ 
\> \> \> Destination $d$ is the point $\epsilon$ before the center 
 of the face on the line \\
\> \> \> from $p_i$ to the center. \\ 
\> \> {\bf Else} \\ 
\> \> \> // $P_1$ is a tetrahedron, a octahedron, or a dodecahedron. \\ 
\> \> \> Select an adjacent face of $P_1$. \\ 
\> \> \> Destination $d$ is the point $\epsilon$ before the center 
 of the face on the line \\
\> \> \> from $p_i$ to the center. \\ 
\> \> {\bf Endif} \\ 
\> \> Move to $d$. \\  
\> {\bf Endif} 
\end{tabbing}
\end{algorithm}

\begin{lemma} 
\label{lemma:show-symm-single} 
Let $P$ be a configuration such that $\gamma(P)$
is a 3D rotation group and $\varrho(P)$ contains neither 
any 3D rotation group nor any symmetry group with
a horizontal mirror plane (except $C_{2h}$).
Let $\{P_1, P_2, \ldots, P_m\}$ be the $\theta(P)$-decomposition
of $P$.
 Then there exists one element of $P_i$ ($1 \leq i \leq m$)
 that forms one of the following polyhedra; 
a regular tetrahedron, a regular octahedron, 
a regular dodecahedron, or an icosidodecahedron.
We further assume that $P_1$ forms one of the above
polyhedra.
Then one step execution of Algorithm~\ref{alg:go-to-center} 
translates $P$ into another configuration $P'$ that satisfies
 (i) $\gamma(P')$ is a 2D rotation group, and 
 (ii) if $\gamma(P') \neq C_1, C_2$, 
 $\theta(P')$ does not have any horizontal mirror plane. 
\end{lemma}

\begin{proof}
 To show the first property of the lemma, 
 we introduce another technique to check
 transitive set of points regarding a symmetry group.\footnote{
 This is an extension of the same technique for rotation groups
 shown in \cite{YUKY16}.}
 Given an arrangement of $G \in {\mathbb S}$ and a seed point $s$
 in an arrangement of $G$,
 by applying all elements of $G$ to $s$, we obtain an 
 orbit of $S = \{g*s ~|~ g \in G\}$.
 Clearly, $S$ is transitive regarding $G$ and the location of $s$
 determines the size of $S$.
 For example,
 when $s$ is on a $k$-fold rotation axis of $G$, $|S| = |G|/k$, 
 when $s$ is on a mirror plane of $G$ but not on a
 rotation axis of $G$, $|S| = |G|/2$, and
 when $s$ is neither on any mirror plane nor on any rotation axis,
 $|S| = |G|$. 

 Since $\gamma(P)$ is a 3D rotation group and
 $\varrho(P)$ does not contain any 3D rotation group,
 there are some robots on the rotation axes of $\gamma(P)$.
 A seed point on a rotation axes of a 3D rotation group
 produces a regular tetrahedron,
 a regular octahedron, a cube, a cuboctahedron, 
 a regular dodecahedron, a regular icosahedron, or 
 an icosidodecahedron.
 However, since $T \prec O$ and $T \prec I$,
 a cuboctahedron and an icosahedron 
 allow $T$ to remain in symmetricity.
 Additionally, a cube allows $C_{4h}$ to remain in symmetricity, and 
 the plane formation is not possible from a
 cubic initial configuration. 
 We have the first property of the lemma. 
  
 Assume that the robots of $P_1$ occupy the points
 $P_1' \subseteq P'$.
 It suffices to show that $\theta(P_1')$ does not have any
 horizontal mirror plane, since 
 $\theta(P')$ is a subgroup of $\theta(P_1')$. 
  
 We first check the rotation group of 
 any resulting configuration of Algorithm~\ref{alg:go-to-center} and
 then we proceed to the combinations of rotation axes and mirror planes. 
 In~\cite{YUKY16}, it has been shown that 
 after one step execution of Algorithm~\ref{alg:go-to-center},
 the rotation group of any resulting configuration is one of the
 rotation groups shown in Table~\ref{table:four-rots}.
 Hence $\theta(P')$ is a composition of these rotation symmetry and
 reflection symmetry. 

 \begin{table}[t]
   \centering
  \caption{Symmetricity and rotation group of the
  four polyhedra and rotation groups after 
  the execution of Algorithm~\ref{alg:go-to-center}}
 \label{table:four-rots}
 \centering
  \begin{tabular}{|l|c|c|c|l|}
   \hline
   Polyhedron of $P$ & $|P|$ & $\gamma(P)$ & $\varrho(P)$ & Candidates of $\gamma(P')$ \\
   \hline
   Regular tetrahedron & 4 & $T$ & $\{D_2, S_4\}$ & $D_2, C_2, C_1$ \\
   \hline
   Regular octahedron & 6 & $O$ & $\{D_3, S_6\}$ & $D_3, C_3, C_1$ \\
   \hline
   Regular dodecahedron & 20 & $I$ & $\{D_5, D_2, S_{10}\}$ &
		   $D_5, D_2, C_5, C_2, C_1$ \\
   \hline
   Icosidodecahedron & 30 & $I$ & $\{S_{10}, S_6\}$ & $C_5, C_3, C_1$ \\
   \hline
  \end{tabular}
 \end{table}

\begin{figure}[t]
 \centering
\subfloat[]{\includegraphics[height=3cm]{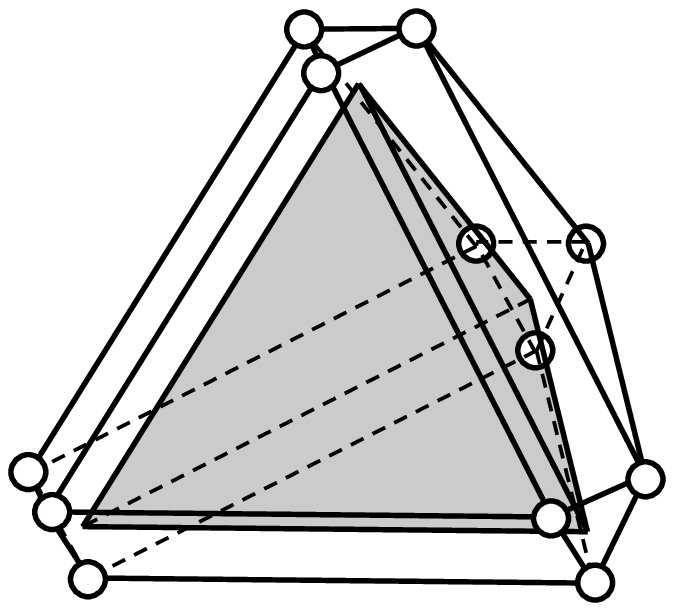}\label{fig:expanded-tetra}}
 \quad 
\subfloat[]{\includegraphics[height=3cm]{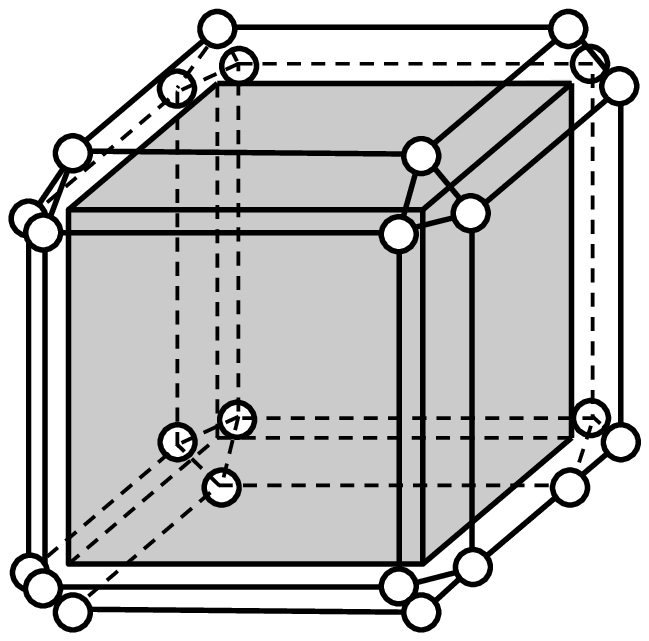}\label{fig:expanded-cube}}
 \\
\subfloat[]{\includegraphics[height=3cm]{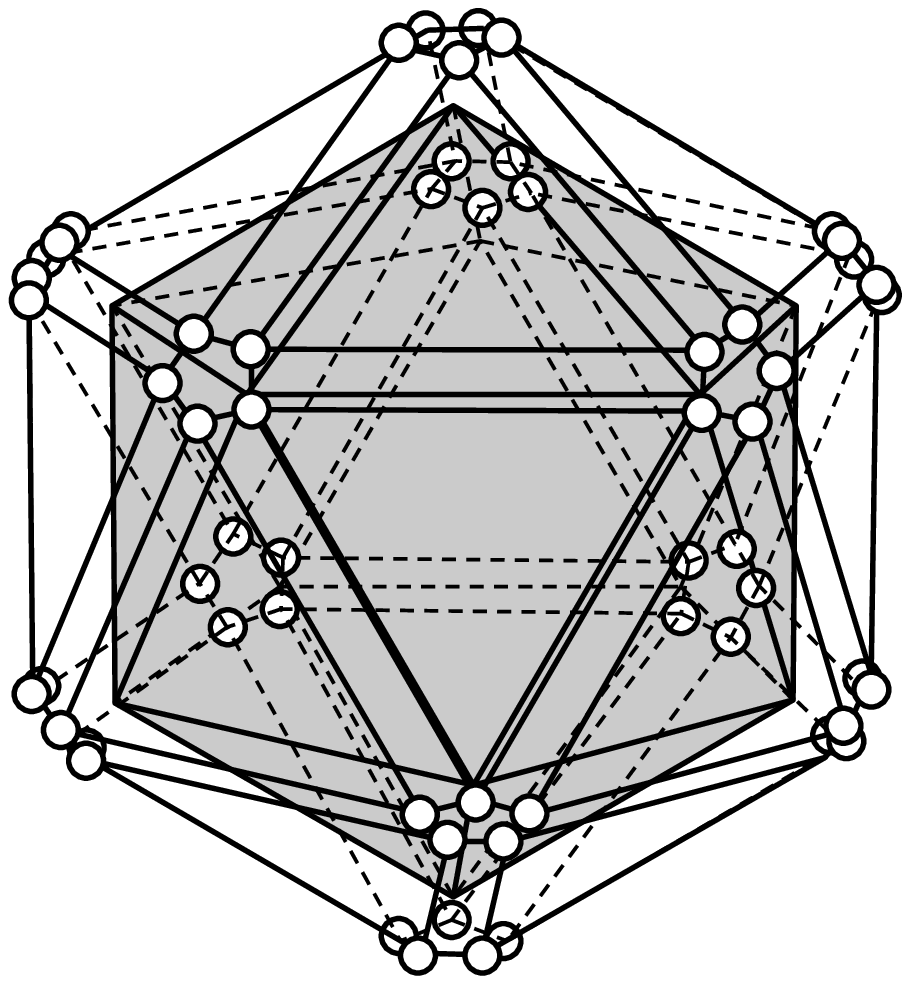}\label{fig:expanded-icosa}}
 \quad 
\subfloat[]{\includegraphics[height=3cm]{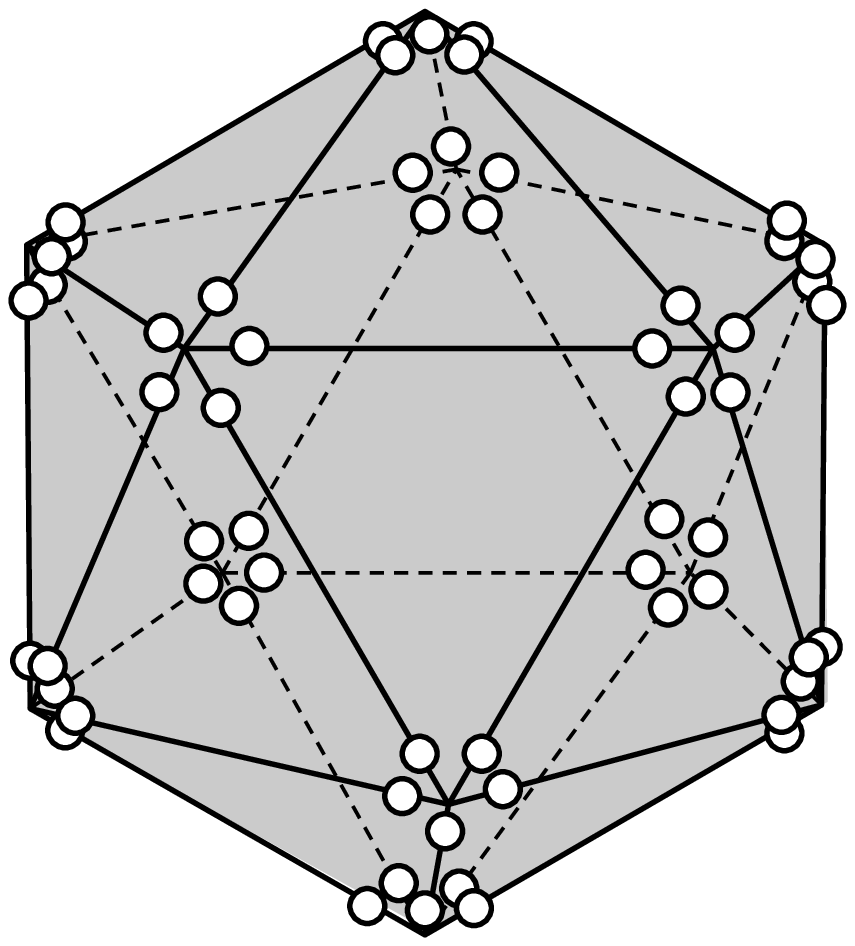}\label{fig:truncated-icosa}}
 \quad 
 \caption{The set of candidate destinations.
 (a) an $\epsilon$-expanded tetrahedron,
 (b) an $\epsilon$-expanded cube,
 (c) an $\epsilon$-expanded icosahedron, and 
 (d) an $\epsilon$-truncated icosahedron. 
 }
\label{fig:dest-candidates}
\end{figure}

 The set of candidate destinations of
 Algorithm~\ref{alg:go-to-center} forms
 the polyhedra shown in Fig.~\ref{fig:dest-candidates}.
 For example, when $P$ forms a regular octahedron,
 possible next positions of the six robots are around the
 centers of the faces, thus, around the vertices of the dual cube.
 Since the robots do not move to the center,
 the candidate destinations form an
 {\em $\epsilon$-expanded cube}, which is obtained by
 expanding the faces of a cube.
 The rotation group of an $\epsilon$-expanded cube 
 is the same as its original polyhedra, i.e., a cube,
 and it is $O$.
 Additionally, its vertices form a
 transitive set regarding $O$. 
 The six robots select a subset of the vertices of this
 $\epsilon$-expanded cube in Algorithm~\ref{alg:go-to-center}. 
 In the same way,
 when $P$ forms a regular tetrahedron,
 the candidate destinations form an
 {\em $\epsilon$-expanded tetrahedron}, 
 when $P$ forms a regular dodecahedron,
 the candidate destinations form an
 {\em $\epsilon$-expanded icosahedron},
 and when $P$ forms an icosidodecahedron,
 the candidate destinations form an
 {\em $\epsilon$-truncated icosahedron}.
 Each of these polyhedra is also transitive (hence spherical)
 regarding the rotation
 group of its original polyhedron. 

 We check the symmetry group of $P_1'$ and depending on $P_1$,
 we have the following four cases.

 \noindent{\bf Case A: When $P_1$ forms a regular tetrahedron.~}
 The set of candidate destinations form an $\epsilon$-expanded
 tetrahedron and $|P_1'| = 4$.
 By Table~\ref{table:four-rots},
 we check whether $\theta(P_1')$ is $D_{2h}$. 

 Assume that $\theta(P_1') = D_{2h}$.
 When the points of $P_1'$ are on the principal axis
 (secondary axes, respectively),
 $P_1'$ is on one plane. 
 When the  points of $P_1'$ are on mirror planes but not on
 any rotation axes, 
 still $P_1'$ is on one plane.
 Otherwise, we have $|P_1'| = |D_{2h}| = 8$,
 and we do not have this case.

 However, since the four robots select one face of a regular
 tetrahedron, $P_1'$ is not on a plane.
 There are following four cases:
 (i) three robots select the same face,
 (ii) two robots select the same face and the remaining two robots
 select another face, 
 (iii) robots are divided into a $2$-set and two $1$-sets
 and the three groups select different faces, and
 (iv) each robot selects different face.
 In any of the four cases, the four robots are not on one plane.
 Hence, we do not have the case where $\theta(P_1')=D_{2h}$. 

 \noindent{\bf Case B: When $P_1$ forms a regular octahedron.~}
 The set of candidate destinations form an $\epsilon$-expanded
 cube and $|P_1'| = 6$.
 By Table~\ref{table:four-rots},
 we check whether $\theta(P_1')$ is $D_{3h}$ or $C_{3h}$. 

 We first show that $P_1'$ is not on a plane. 
 The candidate destinations of one $p \in P_1$ forms a square 
 face of an $\epsilon$-expanded cube. 
 If $P_1'$ is on one plane, say $H$,
 $H$ contains at least one vertex of each square face of
 an $\epsilon$-expanded cube. 
 Clearly, such $H$ does not exist. 

 Assume that $\theta(P_1') = D_{3h}$.
 Thus the points of $P_1'$ are on some mirror planes,
 otherwise we have $|P_1'| = 12$.
 Since $P_1'$ is not on one plane, $P_1'$ forms a triangular prism.
 As Fig.~\ref{fig:triangle-ecube} shows, 
 any regular triangle in an $\epsilon$-expanded cube is centered
 at a point on a $3$-fold rotation axis. 
 Additionally, no combination of these triangles form
 a triangular prism. 
 Hence, we have $\theta(P_1') \neq D_{3h}$. 

\begin{figure}[t]
 \centering
 \includegraphics[height=3cm]{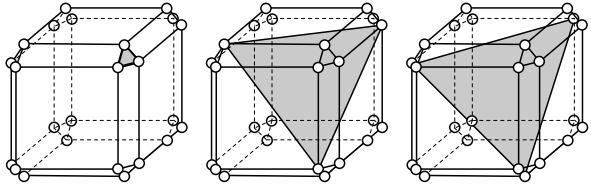}
 \caption{Triangles in an $\epsilon$-expanded cube} 
\label{fig:triangle-ecube}
\end{figure}

 Assume that $\theta(P_1') = C_{3h}$.
 Since $P_1'$ is not on one plane, $P_1'$ is not on the
 horizontal mirror plane of $C_{3h}$ and it forms a triangular prism.
 In the same way as the above discussion, we do not have this case.

 \noindent{\bf Case C: When $P_1$ forms a regular dodecahedron.~}
 The set of candidate destinations form an $\epsilon$-expanded
 icosahedron and $|P_1'| = 20$.
 By Table~\ref{table:four-rots},
 we check whether $\theta(P_1')$ is $D_{5h}$, $D_{2h}$, or $C_{5h}$. 

 We first show that $P_1'$ is not on a plane.
 The candidate destinations of one $p \in P_1$ forms a regular
 triangle face of an $\epsilon$-expanded icosahedron.
 If $P_1'$ is on one plane, say $H$,
 $H$ contains at least one vertex of each regular triangle face of
 an $\epsilon$-expanded icosahedron.
 Clearly, such $H$ does not exist. 
 
 Assume that $\theta(P_1') = D_{5h}$.
 Since $P_1'$ is not on one plane,
 we have the following two cases:
 (a) $P_1'$ contains a pentagonal prism (size $10$), and
 (b) $P_1'$ contains a transitive $20$-set regarding $D_{5h}$,
 (c) $P_1'$ contains a set of points on the principal
 rotation axis of $D_{5h}$. 
 As Fig.~\ref{fig:pentagon-eicosa} shows, 
 any regular pentagon in an $\epsilon$-expanded icosahedron is centered
 at a point on a $5$-fold rotation axis. 
 Additionally, no combination of these pentagons form
 a pentagonal prism. 
 Hence, we do not have case (a).

\begin{figure}[t]
 \centering
 \includegraphics[width=12cm]{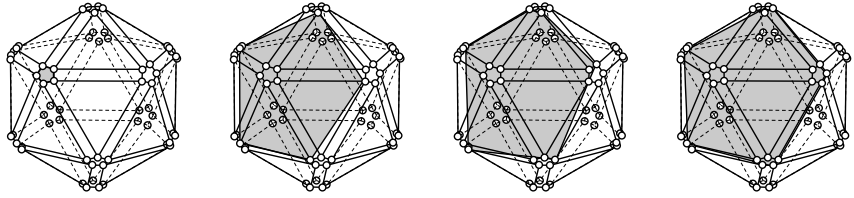}
 \caption{Pentagons in an $\epsilon$-expanded icosahedron} 
\label{fig:pentagon-eicosa}
\end{figure}

 Any transitive $20$-set regarding $D_{5h}$ consists of
 two pentagonal prisms.
 From the above discussion, we do not have case (b).
 
 As sown in Fig.~\ref{fig:pentagon-eicosa},
 any pentagon in an $\epsilon$-expanded icosahedron 
 has no point above its rotation axis because
 there is no point of an $\epsilon$-expanded icosahedron
 on its $5$-fold rotation axis.
 Thus we do not have case (c) and we have $\theta(P_1') \neq D_{5h}$.
 
 Assume that $\theta(P_1') = C_{5h}$.
 Since $P_1'$ is not on one plane,
 we have the following two cases:
 (d) $P_1'$ contains a pentagonal prism (size $10$), and
 (e) $P_1'$ contains a set of points on the principal axis.
 In the same way as above discussion, we have $\theta(P_1') \neq C_{5h}$.

 Assume that $\theta(P_1') = D_{2h}$.
 The size of a transitive set of points regarding $D_{2h}$ is
 either $8$, $4$ (on a mirror plane), or $2$ (on a rotation axis). 
 Since $P_1'$ is not on one plane,
 $P_1'$ does not consist of transitive $4$-sets.
 When $P_1'$ contains a transitive $2$-set,
 the number of transitive $2$-sets is greater than one
 because $|P_1'|=20$.
 However, since an $\epsilon$-expanded icosahedron is spherical,
 at most two points of it are on a line.
 Thus $P_1'$ should contain a transitive $4$-set because of its size.

 Fig.~\ref{fig:cuboid-eicosa} shows all possible cuboids
 in an $\epsilon$-expanded  icosahedron.
 Then their mirror planes contain a $2$-fold axis and
 two $5$-fold axes, in other words,
 four vertices of the original icosahedron. 
 Since a vertex of a regular icosahedron is broken into
 five points in an $\epsilon$-expanded icosahedron,
 there is no rectangle containing the mirror plane of
 any of such cuboids.
 Hence, we have $\theta(P_1') \neq D_{2h}$. 

\begin{figure}[t]
 \centering
 \includegraphics[width=12cm]{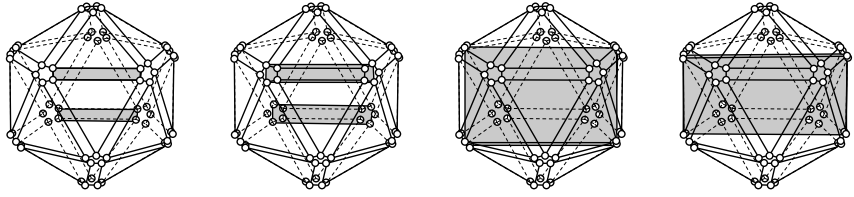}
 \caption{Cuboids in an $\epsilon$-expanded icosahedron.
 The figures show two parallel faces of each possible cuboid. } 
\label{fig:cuboid-eicosa}
\end{figure}

 \noindent{\bf Case D: When $P_1$ forms a regular icosidodecahedron.~}
 The set of candidate destinations form an $\epsilon$-truncated 
 icosahedron and $|P_1'| = 30$.
 By Table~\ref{table:four-rots},
 we check whether $\theta(P_1')$ is $C_{5h}$ or $C_{3h}$.  

 We first show that $P_1'$ is not on a plane.
 The candidate destinations of one $p \in P_1$ forms an edge 
 of an $\epsilon$-truncated icosahedron.
 If $P_1'$ is on one plane, say $H$,
 $H$ contains at least one endpoint of each edge of
 an $\epsilon$-expanded icosahedron.
 Clearly, such $H$ does not exist. 

 Assume that $\theta(P_1') = C_{5h}$.
 Since $P_1'$ is not on one plane,
 we have one of the following two cases:
 (a) $P_1'$ contains a pentagonal prism, or (b) $P_1'$ contains a
 set of points on the rotation axis of $C_{5h}$.
 As Fig.~\ref{fig:pentagon-ticosa} shows, 
 any regular pentagon in an $\epsilon$-truncated icosahedron is centered
 at a point on a $5$-fold rotation axis. 
 Additionally, no combination of these pentagons form
 a pentagonal prism. 
 Hence, we do not have case (a).

\begin{figure}[t]
 \centering
 \includegraphics[width=12cm]{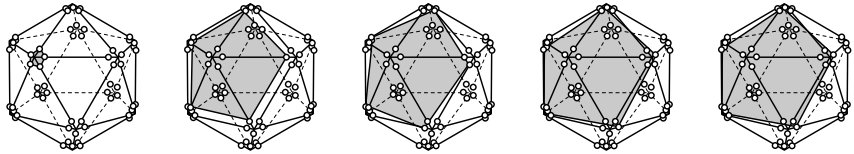}
 \caption{Pentagons in an $\epsilon$-truncated icosahedron} 
\label{fig:pentagon-ticosa}
\end{figure}

 As sown in Fig.~\ref{fig:pentagon-ticosa},
 any pentagon in an $\epsilon$-truncated icosahedron 
 has no point above its rotation axis because
 there is no point of an $\epsilon$-truncated icosahedron
 on its $5$-fold rotation axis.
 Thus we do not have case (b) and we have $\theta(P_1') \neq C_{5h}$.

 Assume that $\theta(P_1') = C_{3h}$.
 Since $P_1'$ is not on one plane,
 we have one of the following two cases:
 (c) $P_1'$ contains a triangular prism, or (d) $P_1'$ contains a
 set of points on the rotation axis of $C_{3h}$.
 As Fig.~\ref{fig:triangle-ticosa} shows, 
 any regular triangle in an $\epsilon$-truncated icosahedron is centered
 at a point on a $3$-fold rotation axis. 
 Additionally, no combination of these triangles form
 a triangular prism. 
 Hence, we do not have case (c).

\begin{figure}[t]
 \centering
 \includegraphics[width=12cm]{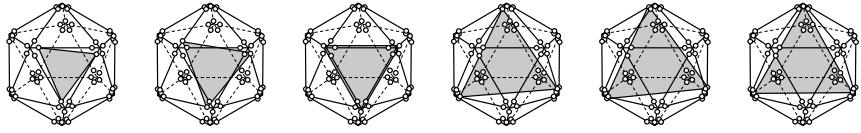}
 \caption{Triangles in an $\epsilon$-truncated icosahedron} 
\label{fig:triangle-ticosa}
\end{figure}

 As sown in Fig.~\ref{fig:pentagon-ticosa},
 any regular triangle in an $\epsilon$-truncated icosahedron 
 has no point above its rotation axis because
 there is no point of an $\epsilon$-truncated icosahedron
 on its $3$-fold rotation axis.
 Thus we do not have case (d) and we have $\theta(P_1') \neq C_{3h}$.
 
 From the above four cases, we conclude that 
 $\theta(P_1')$ is a 2D rotation group and
 if $\theta(P_1')$ have at least one rotation axis,
 it does not have a horizontal mirror plane (except $C_{2h}$). 
 Since $\theta(P')$ is a subgroup of $\theta(P_1')$,
 we have the lemma. 
 \qed 
\end{proof}

We finally note that the robots cannot remove 
vertical mirror planes of $P$ with the go-to-center algorithm. 
For example, consider an initial configuration $P$
where the robots form a regular octahedron.
Thus $\theta(P) = O_h$.
Consider an embedding of $D_{3 v}$ in $\theta(P)$. 
The $3$-fold rotation axis of $D_{3v}$ overlaps a 
$3$-fold rotation axis of $O_h$ and each vertical
mirror plane of $D_{3v}$ contains two trajectories of the go-to-center
algorithm. Actually, the six robots can take these trajectories and 
the resulting configuration $P'$ forms a
triangular anti-prism (thus $\theta(P')=D_{3v}$).

\subsection{Landing Algorithm}
\label{subsec:landing}

In this section, we show a plane formation algorithm
for Type 2 and Type 3 initial configurations. 
When $\gamma(P)$ of a current configuration is a
cyclic group or a dihedral group, 
our basic strategy is to make the robots agree on the
plane perpendicular to the principal axis
and containing $b(P)$ and then we send the robots to
the plane. 
Each robot moves along a perpendicular to the plane. 
To avoid multiplicities, we need some tricks for the
following two cases: 
First, when $\theta(P)$ has a horizontal mirror plane, 
the condition of Theorem~\ref{theorem:main} guarantees that
there is at least one element of the $\theta(P)$-decomposition
of $P$ on it.
To remove this mirror plane,
we first make the robots of such an element leave their
current positions (Fig.~\ref{fig:intro-remove-mirror}). 
The other case is when $\gamma(P)$ is a dihedral group
and some element, say $P_i$,
of the $\theta(P)$-decomposition of $P$
has a horizontal mirror plane. 
Since the target plane is this mirror plane, 
the final destination of any symmetric two robots are the same. 
However, Theorem~\ref{theorem:main} guarantees that
there exists at least one element, say $P_j$,
of the $\theta(P)$-decomposition
of $P$ that does not have a horizontal mirror plane.
The robots of $P_i$ use $P_j$ to break their symmetric 
landing points (Fig.~\ref{fig:intro-remove-col}). 

\begin{figure}[t]
 \centering
\subfloat[]{\includegraphics[height=4cm]{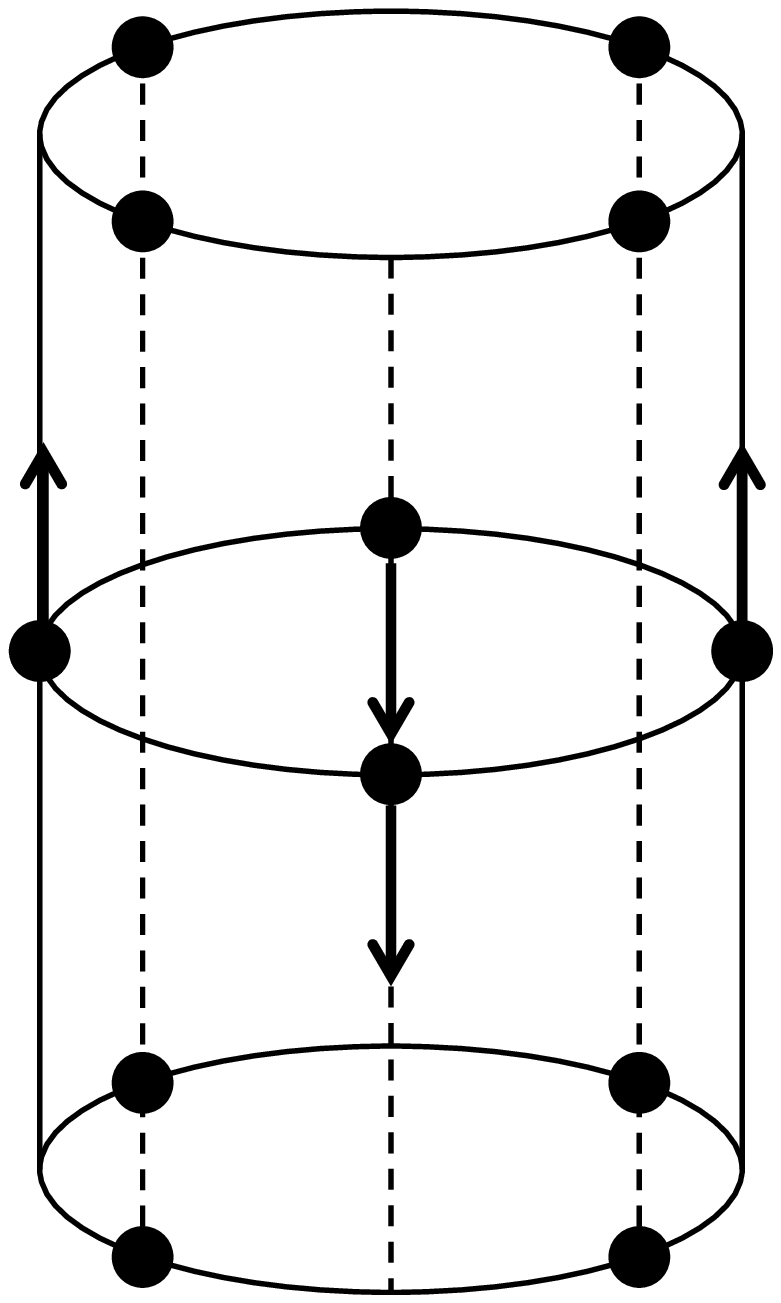}\label{fig:intro-remove-mirror}}
 \quad 
\subfloat[]{\includegraphics[height=4cm]{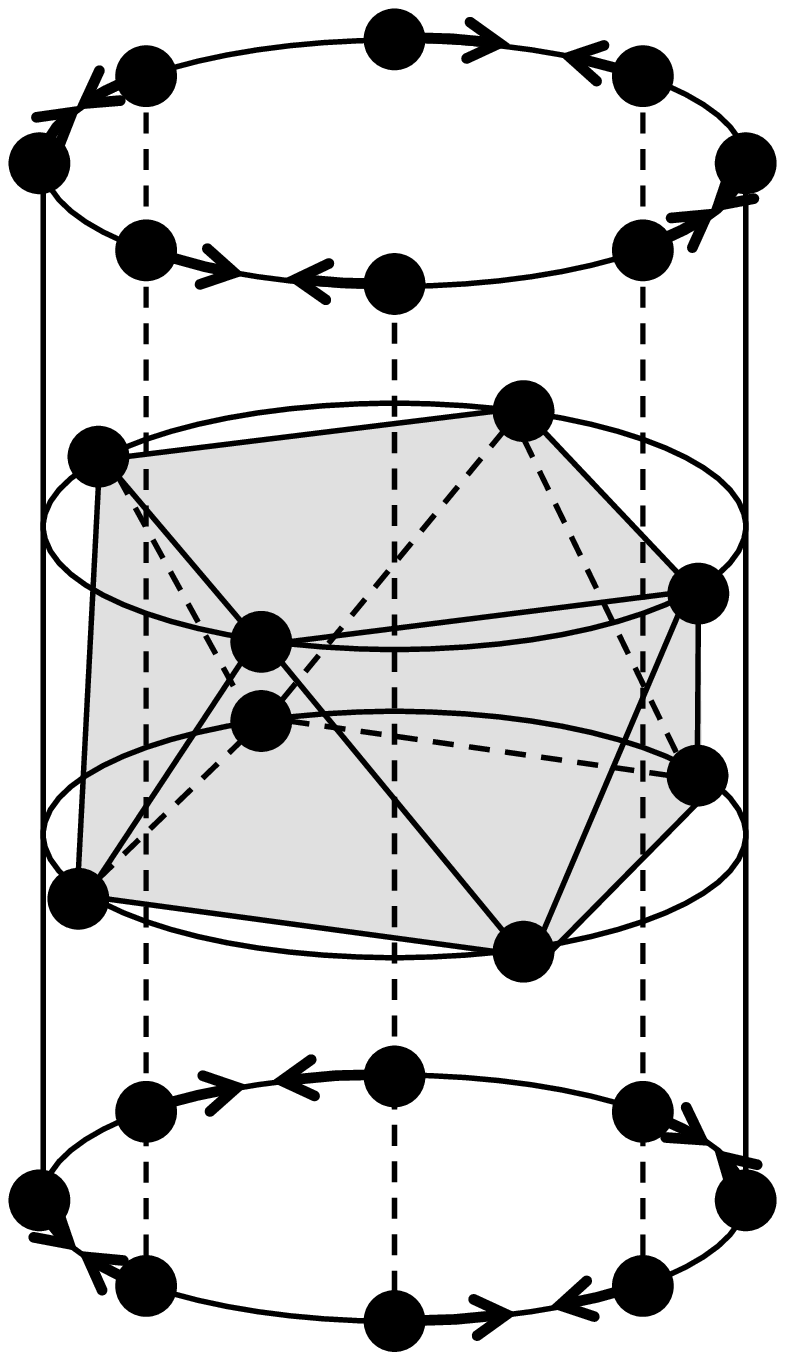}\label{fig:intro-remove-col}}
 \caption{Removing horizontal mirror plane and resolving collisions.
 (a) The robots on the horizontal mirror plane move vertically,
 and removes the plane.
 (b) The robots forming a prism rotates by using an anti-prism. }
\label{fig:intro-remove}
\end{figure}

The proposed algorithm consists of five phases.
The first three phases break the mirror plane of $\theta(P)$ and
resolves the collisions on the target plane. 
The fourth phase makes the robots agree on the
target plane and in the fifth phase
each robot computes the destinations of all robots
to avoid any collision. 
The fourth and the fifth phases are done in
local computation at each robot. 
Finally, the robots move to their final destinations 
in the same cycle.
In any configuration $P$, the robots execute the algorithm
for the smallest phase number. 
The robots can easily agree on which phase to execute
because the condition for each of the five phases
divide the set of all configurations with 2D rotation groups
into disjoint subsets. 
Depending on the execution, some phases may be skipped.
Since the proposed algorithm is designed for the oblivious robots,
it is always described for a current configuration.

\subsubsection{First Phase: Removing the Mirror Plane}

By the condition of Theorem~\ref{theorem:main}, 
when $\gamma(P)$ is a cyclic group or a dihedral group and 
has a horizontal mirror plane,
the mirror plane contains some robots. 
The preparation step guarantees that this element is $P_1$
of the $\theta(P)$-decomposition of $P$.
Let $k$ be the folding of the principal axis of $\theta(P)$. 
Intuitively, the first phase makes the robots select the upward
direction or the downward direction regarding this mirror plane
and the robots move to the selected directions.
Any resulting configuration does not have the horizontal mirror plane
any more because for each new positions of the robots, 
there is no corresponding point regarding the horizontal mirror plane. 
However, for the simplicity of the correctness proof,
these next positions are selected more carefully. 

The robots consider a fictitious prism with 
a regular $k$-gon base inscribed in $I(P)$,
that share the horizontal mirror plane (Fig.~\ref{fig:fict-prism}).
However, the size is selected so that
the length of the edge of the regular $k$-gon base 
is one tenth of the length of its side edge, and
its arrangement is determined so that 
the plane formed by $b(P)$ and a side edge contains
a point of $P_1$. 
Then, each $p \in P_1$ moves toward one of the nearest
vertex of this fictitious prism.

\begin{figure}[t]
 \centering
 \includegraphics[height=3cm]{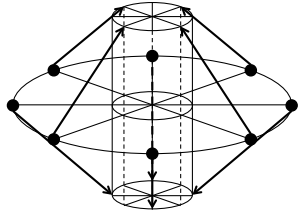}
 \caption{Fictitious prism for the first phase.} 
\label{fig:fict-prism}
\end{figure}

\begin{lemma}
 \label{lemma:remove-mirror}
 Let $P$ be a configuration such that $\gamma(P)$
 is a 2D rotation group ($\neq C_1, C_2$) and some robots are
 one the horizontal mirror plane of $\theta(P)$. 
 In this configuration, the robots execute the first phase,
 and a new configuration $P'$ yields. 
 Then $P'$ satisfies one of the following conditions:
 (a) $\gamma(P')=C_1$,
 (b) $\gamma(P')= C_2$, or 
 (c) $\theta(P')$ does not have any horizontal mirror plane. 
\end{lemma}

\begin{proof} 
 Let $\{P_1, P_2, \ldots, P_m\}$ and $k$
 be the $\theta(P)$-decomposition of $P$ and
 the folding of the principal axis of $\theta(P)$. 
 We denote the vertices of the fictitious prism by $C$.
 Clearly, $\gamma(C) = D_k$. 
 During the transition from $P$ to $P'$, 
 the robots of $P_1$ select a subset $C' \subset C$
 and moves to the selected vertices.
 Other robots of $P \setminus P_1$ do not move.
 Thus clearly the mirror plane of $\theta(P)$ is not a
 mirror plane of $\theta(P')$ (even $\theta(C')$)
 because each $c \in C'$ does not have the corresponding point
 regarding this initial mirror plane. 

 We separately consider the following two cases. 
 First, when $\theta(P\setminus P_1) = \theta(P)$, 
 $\theta(P' \setminus C') = \theta(P)$ and 
 $\theta(P')$ is a subgroup of $\theta(P)$. 
 In other words, the symmetry group of $P$ is kept by
 the robots of $P \setminus P_1$ 
 because these robots do not move during the transition from $P$ to
 $P'$, the symmetry group that acts on them does not change. 
 Since $\theta(P')$ is a subgroup of $\theta(C')$ and
 $\theta(P' \setminus C') = \theta(P)$, 
 we check $\theta(C')$. 

 We first consider rotation axes of $\theta(C')$.
 We have the following four cases:

 \noindent{\bf Case A:~}
 If the principal axis of $\theta(C)$ remains as
 some rotation axis of $\gamma(C')$,
 its folding is a divisor of $k$ because of the
 movement of the robots of $P_1$. 
 Clearly, $\gamma(C')$ has no horizontal mirror plane from the
 above discussion. 

 \noindent{\bf Case B:~}
 If a $2$-fold axis of $\theta(C)$ remains as
 some ration axis of $\gamma(C')$,
 its folding remains two because any
 subset of $C$ that is on a plane perpendicular to
 this $2$-fold axis forms a line or a rectangle because the prism is
 long. Actually, we do not have the case of a square
 because $C'$ is not symmetric regarding the 
 horizontal mirror plane of $C$.

 \noindent{\bf Case C:~}
 If a new rotation axis appears and it has an intersection with the
 top or the base of $C$, 
 it is a $2$-fold axis because any
 subset of $C$ that is on a plane perpendicular to
 this axis forms a line.

 \noindent{\bf Case D:~}
 If a new rotation axis appears and it has an intersection with the
 side face of $C$, 
 it is a $2$-fold axis because any
 subset of $C$ that is on a plane perpendicular to
 this axis forms a line.

 From the above four cases,
 $\theta(C')$ is neither $D_{\ell h}$ nor $C_{\ell h}$ for
 any $\ell > 2$ because
 the possible principal axis do not have any mirror plane
 by Case A.
 The remaining case is $D_{2h}$. 
 The rotation axes of Case A and Case B do not form
 $D_{2h}$ because there is no horizontal mirror plane.
 Even when the rotation axes of Case C and Case D form $D_{2h}$, 
 this $D_{2h}$ does not act on $P' \setminus C'$
 since  $\theta(P' \setminus C') = \theta(P)$.
 Hence, $\theta(P')$ is $C_{2h}$ or does not have any horizontal mirror
 planes. 

 Second, we consider the case where
 $\theta(P \setminus P_1) \neq \theta(P)$.
 From the preparation phase,
 all points of $P$ are on the horizontal mirror plane
 (thus the plane formation is finished),
 or all points of $P$ are on the principal axis.
 In the first case, the assumption
 $\theta(P\setminus P_1) = \theta(P)$
 is used in the combination of the rotation axes of Case C and Case D.
 Thus what we should check is this case.
 When the rotation axes of Case C and Case D form
 $D_{2h}$, it should act on the points on the original
 principal axis of $\theta(P)$.
 However, they do not act on these points since they do
 not overlap these points. 
 Hence, $\theta(P')$ is $C_{2h}$ or does not have any horizontal mirror
 planes.
 \qed
\end{proof}

In the following, we assume that when the current configuration
$P$ has a rotation axis, it does not have any horizontal mirror plane
for the principal axis or
$\theta(P) = C_{2h}$.  

\subsubsection{Second Phase: Collision Avoidance for Dihedral Groups}

When $\gamma(P)$ is a dihedral group, say $D_{\ell}$ or $D_{\ell v}$, 
and the $\theta(P)$-decomposition of $P$ contains a prism,
our basic strategy makes multiplicities. 
Let $P_i$ be the element of the $\theta(P)$-decomposition of $P$
that forms a prism (Fig.~\ref{fig:intro-remove-col}). 
Since $\theta(P)$ is $D_{\ell}$ or $D_{\ell v}$, 
any side edge of $P_i$'s prism contains a $2$-fold axis. 
Remember that the vertical mirror planes of $D_{\ell v}$ do not
contain any $2$-fold axis.
Since $\theta(P)$ does not have a horizontal mirror plane,
there exists at least one element $P_j$
that does not form such a prism.
Let $j$ be the minimum index among such elements.
Then for each robot $p \in P_i$,
the nearest point of $P_j$ is uniquely determined.
If there are nearest two points of $P_j$ for
$p \in P_i$, then $P_j$ is not transitive regarding $\theta(P)$
as shown in
Fig.~\ref{fig:nearest-points}: 
When $\theta(P)=D_{\ell v}$,
if $P_j$ is not on any mirror plane of $\theta(P)$ and
each $p \in P_i$ has two nearest points of $P_j$, 
then the mirror planes produce $8 \ell$ points,
which means $P_j$ is not transitive (Fig.~\ref{fig:Dv-1}).
Otherwise, the nearest points are on the mirror planes
of $\theta(P)$ and $|P_j|$ should be $2 \ell$,
but the $2$-fold axes produces $4 \ell$ points,
a contradiction (Fig.~\ref{fig:Dv-2}).
We can show the property for the case of $\theta(P) = D_{\ell}$
in the same argument (Fig.~\ref{fig:D-1} and \ref{fig:D-2}). 

\begin{figure}[t]
 \centering
\subfloat[$D_{3v}$]{\includegraphics[height=2cm]{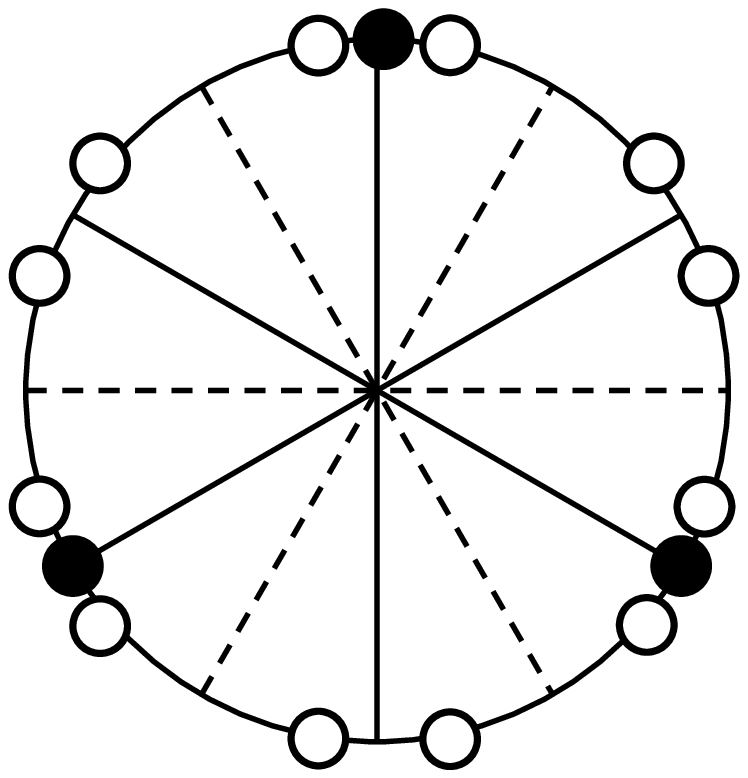}\label{fig:Dv-1}}
 \quad 
\subfloat[$D_{3v}$]{\includegraphics[height=2cm]{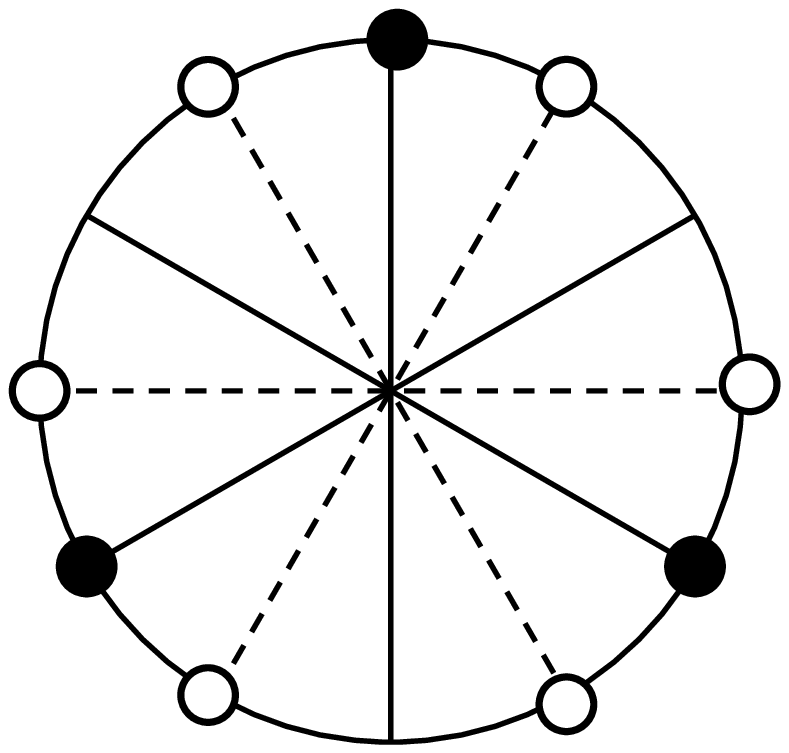}\label{fig:Dv-2}}
 \quad 
\subfloat[$D_3$]{\includegraphics[height=2cm]{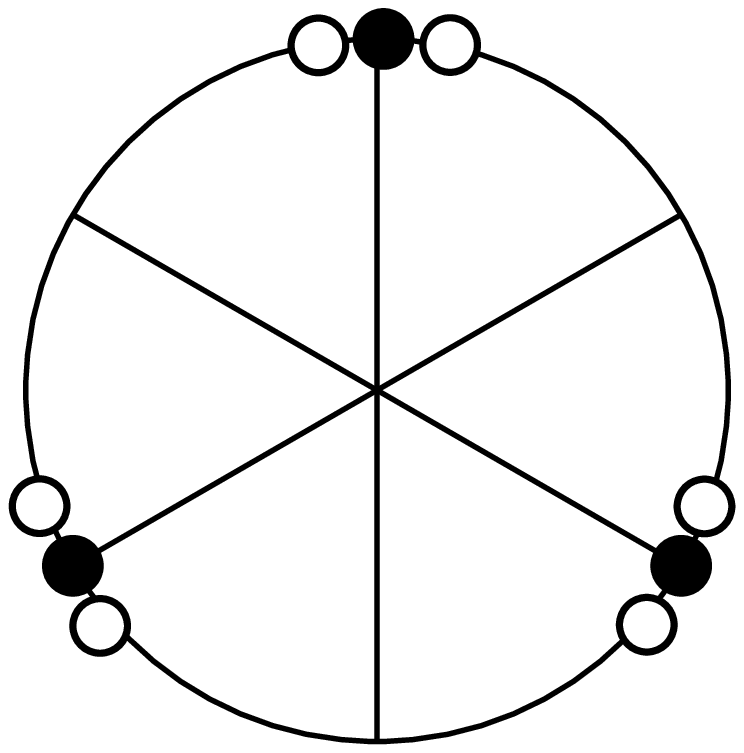}\label{fig:D-1}}
 \quad 
\subfloat[$D_3$]{\includegraphics[height=2cm]{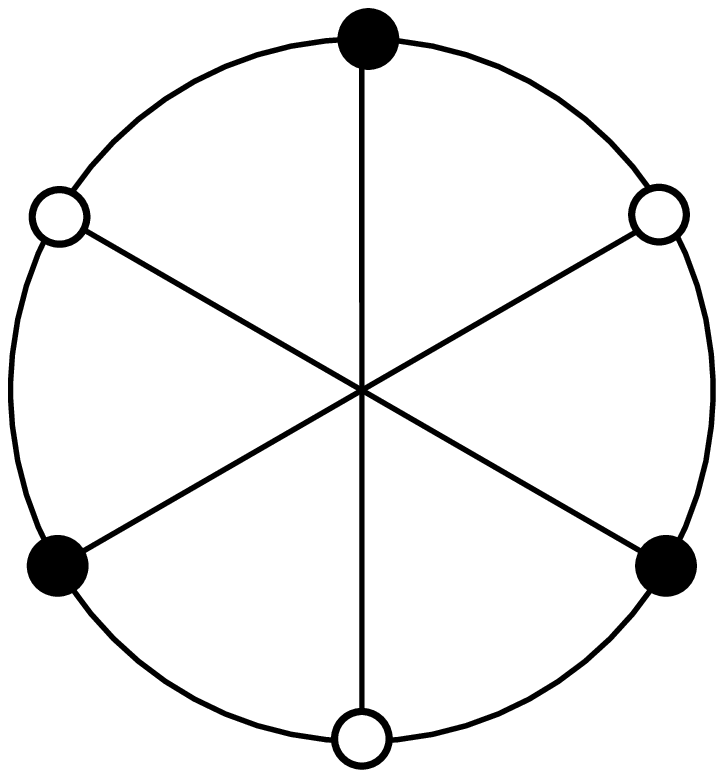}\label{fig:D-2}}
 \quad 
 \caption{Examples of multiple nearest points.
 The figures show the top (or the base) of a prism. Lines show $2$-fold axes
 and broken lines show vertical mirror planes. The principal axis passes
 the center of the circle and the black circles are the point forming a
 prism. The white circles are examples of two nearest points for each
 black circle.
 (a) When a white circle is not on a mirror plane,
 $D_{3v}$ produces $24$ points. 
 (b) When a white circle is on a mirror plane,
 $D_{3v}$ produces $12$ points.
 (c) When a white circle is on a mirror plane,
 $D_3$ produces $12$ points.
 (d) When a white circle is on a rotation axes, 
 $D_3$ produces a prism. }
\label{fig:nearest-points}
\end{figure}

This selection gives an agreement of direction to the robots
of $P_i$. 
The robots of $P_i$ circulates toward the nearest point
of $P_j$ and twist their prism. 
This movement resolves the collisions among
the perpendiculars from the robots of $P_i$
to the target plane. 

\subsubsection{Third Phase: Collision Avoidance on the Principal
Axis}

When $\gamma(P)$ is a dihedral group and some robots are on the
principal axis,
we need another trick to resolve the collisions of these robots.
Clearly,
these robots form element(s) of 
the $\theta(P)$-decomposition of $P$ and
the size of each of such elements $P_k$ is two. 

We also use an element of $P_j$ that forms a ``twisted'' prism 
in the same way as the previous case. 
Each point of $p \in P_k$ selects the nearest point of $P_j$, 
however in this case, the robots of $P_k$ do not move. 
Other robots consider the vertices of the twisted prism
as possible destinations of this fictitious move. 

\subsubsection{Fourth Phase: Agreement of the Target Plane} 

The robots agree on the target plane. 
Depending on $\theta(P)$ of the current configuration $P$,
we have the following five cases: 
When $\gamma(P)$ is a cyclic group or a dihedral group, 
the robots agree on the plane perpendicular to the
principal axis and containing $b(P)$.
The exceptional case is when $\theta(P) = C_{2h}$.
In this case, since the robots are not on one plane,
there exists at least one element $P_i$
of the $\theta(P)$-decomposition of $P$ such that $|P_i| = 4$.
Let $i$ be the minimum index among such elements.
Then, the robots agree on the plane formed by $P_i$.
We note that $P_i$ forms a rectangle perpendicular to the
horizontal mirror plane and the agreed plane
is not a mirror plane for $P$. 

When $\theta(P)$ is a rotation-reflection, 
the robots agree on the mirror plane of $\theta(P)$. 

When $\theta(P)$ is a bilateral symmetry, 
the size of each element of the
$\theta(P)$-decomposition of $P$ is one or two. 
Since the robots are not on one plane,
there exists at least one element $P_i$ such that
$|P_i| = 2$ and forms a perpendicular line regarding the
mirror plane. 
If there is just one such element $P_i$, 
the robots agree on the plane defined by $P_i$ (a line)
and $P_1$ (a point).
If $i=1$, then the algorithm uses $P_2$ instead of $P_1$. 
Otherwise, let $P_i$ and $P_j$ be the elements with the
minimum and second-minimum index of such elements.
Then the robots agree on the plane defined by these two
elements. 
The selected plane is not a mirror plane for $P$.

When $\theta(P)$ is a central inversion, 
the size of each element of the
$\theta(P)$-decomposition of $P$ is one or two. 
Since the robots are not one one plane,
there are more than one elements of size two.
Each of these elements form a line and
they all intersect at the center of inversion. 
Let $P_i$ and $P_j$ be such elements
with the minimum and second-minimum index. 
Then the robots agree on the plane defined by these two elements. 

When $\theta(P)$ is $C_1$, by Lemma~\ref{lemma:ordering}, 
the robots can agree on the total ordering of themselves and 
agree on the plane defined by $P_1$, $P_2$, and $P_3$.

\subsubsection{Fifth Phase: Computation of Final Positions} 

Let $P$ and $\{P_1, P_2, \ldots, P_m\}$
be the current configuration and its 
$\theta(P)$-decomposition. 
Through the previous four phases,
for each element $P_i$ ($i=1, 2, \ldots, m$),
the foots of perpendiculars from the points of $P_i$ to the
target plane do not overlap. 

The robots land on the agreed plane along a perpendicular
to it.
However, to avoid collisions among different elements of the
$\theta(P)$-decomposition of $P$, 
each robot computes the destinations of all other robots. 
The computation of destinations starts from
$P_1$ and 
proceeds to $P_2, P_3, \ldots$.
When the destination of robot $p \in P_i$ is already
a destination of another robot $q \in P_j$ ($j < i$),
$p$ first computes the largest circle $C$ 
that is centered at the destination and does not contain
any other destinations of the robots in $P_1, P_2, \ldots, P_{i-1}$
except the center.
Then, let $C'$ be the concentric circle of $C$ whose radius is
$1/2^{i}$. Robot $p$ selects an arbitrary point of $C'$
as its destination and $C'$ is considered as (possible) destination
of $p$ in the subsequent computation. 
Finally, the robots move to their destinations
in the same cycle. 

As the non-oblivious robots can execute the proposed algorithm, 
we have the following theorem, that
together with Theorem~\ref{theorem:nec}, proves our main theorem. 
\begin{theorem}
 \label{theorem:suf}
 Irrespective of obliviousness,
 the FSYNC robots without chirality can form a plane
 from an initial configuration $P$ 
 if $\varrho(P)$ consists of
 $C_1$, $C_s$, $C_i$, $C_k$, $C_{k v}$, $C_{2h}$,
 $D_{\ell}$, $D_{\ell v}$, and $S_m$. 
\end{theorem}

\section{Conclusion and Discussion}
\label{sec:concl}

We considered the plane formation problem
by FSYNC robots without chirality. 
We extended the notion of symmetricity in \cite{YUY16}
to the composition of rotation symmetry and
reflection symmetry. 
We gave a characterization of initial configurations
from which the FSYNC robots without chirality can form
a plane.
We then showed a plane formation algorithm
for oblivious FSYNC robots without chirality 
because existing plane formation algorithm 
does not work correctly for our robots.

\bibliographystyle{plain}
\bibliography{./bibitems}

\newpage 
\appendix 

\section{Seventeen Symmetry Groups in 3D-space }
\label{sec:list-of-symmetries}

We summarize the seventeen symmetry groups in 3D-space.
Each rotation group is determined by rotation axes,
mirror planes, and their arrangement.
Our results also heavily rely on the order of
each symmetry group and horizontal mirror planes.
The following three tables shows these properties
together with typical polyhedra obtained by
a seed point in each symmetry group.

\begin{table}[h]
 \centering
 \caption{Symmetry groups without rotation axis}
 \begin{tabular}{|l|l|c|}
  \hline 
   & Symmetry & Order \\
  \hline
  $C_1$ & Identity element & 1 \\
  \hline 
  $C_i$ & Point of inversion & 2 \\
  \hline
  $C_s$ & Single mirror plane & 2 \\
  \hline
 \end{tabular}
\end{table}

\begin{table}[h]
 \centering
 \caption{Symmetry groups with 2D rotation groups}
 \begin{tabular}{|l|c|c|c|c|l|}
  \hline 
  & \multirow{2}{*}{Principal axis} & \multirow{2}{*}{Other axes} &
     Horizontal & \multirow{2}{*}{Order} &
     \multirow{2}{*}{Example polyhedra} \\
  & & & mirror & & \\ 
  \hline 
  $C_k$ & $k$-fold & - & N & $k$ & \\
  \hline 
  $C_{k h}$ & $k$-fold & - & Y & $2k$ & \\
  \hline 
  $C_{k v}$ & $k$-fold & - & N & $2k$ & Pyramid with regular $k$-gon base\\
  \hline 
  $D_{\ell}$ & $\ell$-fold & $\ell$ 2-fold axes & N & $2 \ell$ & \\ 
  \hline 
  $D_{\ell h}$ & $\ell$-fold &  $\ell$ 2-fold axes & Y & $4 \ell$ &
		      Hexagonal prism for $D_6$ \\
  \hline 
  $D_{\ell v}$ & $\ell$-fold &  $\ell$ 2-fold axes & N & $4 \ell$ &
		      Hexagonal anti-prism for $D_{6v}$\\ 
  \hline 
  $S_m$ & $m$-fold & - & Y & $m$ &
		      Hexagonal anti-prism for $S_{12}$\\ 
  \hline 
 \end{tabular}
\end{table}

\begin{table}[h]
   \centering
 \caption{Symmetry groups with 3D rotation groups}
 \begin{tabular}{|l|c|c|c|c|c|c|l|}
  \hline 
  & \multicolumn{4}{c|}{Rotation axes (folding)} &
  \multirow{2}{*}{Mirror planes} &
  \multirow{2}{*}{Order} &
  \multirow{2}{*}{Example polyhedra} \\
  \cline{2-5} 
  & 2 & 3 & 4 & 5 & & & \\ 
  \hline 
  $T$   & 3 & 4 & - & - & 0 & 12 & Snub tetrahedron \\
  \hline 
  $T_d$ & 3 & 4 & - & - & 6 & 24 & Regular tetrahedron\\
  \hline 
  $T_h$ & 3 & 4 & - & - & 3 & 24 & \\
  \hline 
  $O$   & 6 & 4 & 3 & - & 0 & 24 & Snub cube\\
  \hline 
  $O_h$ & 6 & 4 & 3 & - & 9 & 48 & Cube, regular octahedron\\
  \hline 
  $I$   & 15 & 10 & - & 6 & 0 & 60 & Snub icosahedron\\
  \hline 
  $I_h$ & 15 & 10 & - & 6 & 15 & 120 & Regular dodecahedron \\
  \hline 
 \end{tabular}
 \end{table}

Fig.~\ref{fig:O-mirrors} shows the 
arrangement of the rotation axes and the mirror planes
of $D_{3h}$, $D_{3v}$, $O_h$, and $O_v$. 
The arrangement of the mirrors of $O_h$ is related to that of
$T_h$ and $T_d$.

\begin{figure}[t]
 \centering
\subfloat[Mirrors of $D_{3h}$]{\includegraphics[height=2cm]{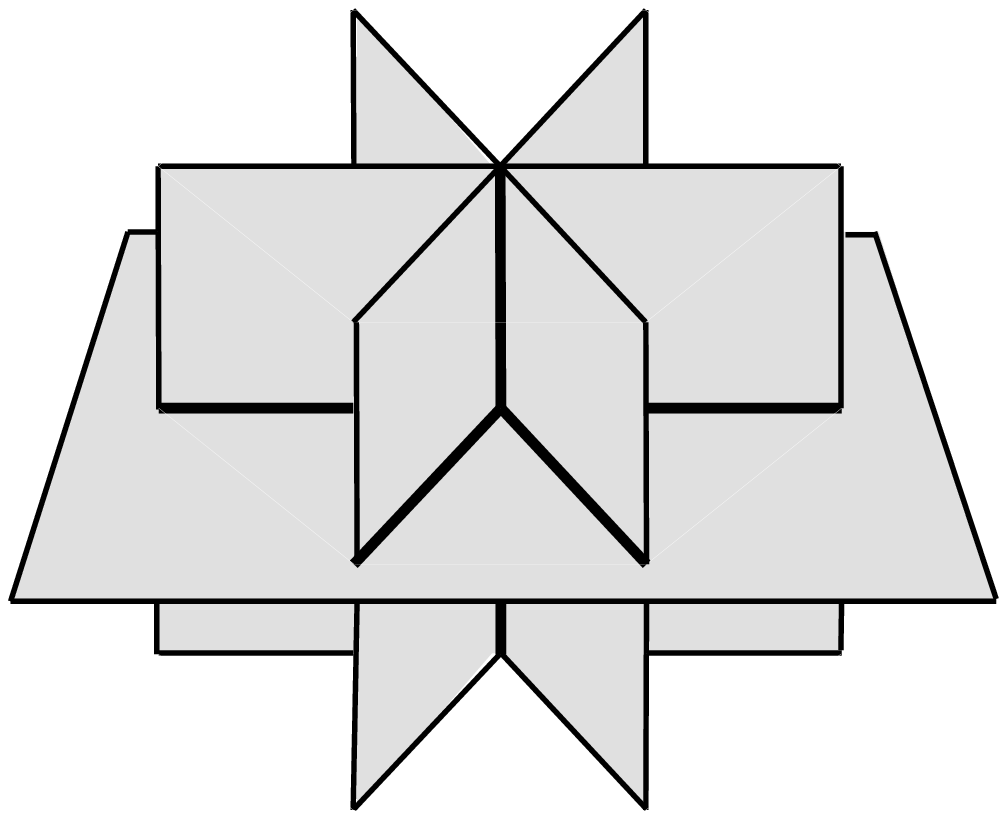}\label{fig:Dh-mirrors}}
 \quad 
\subfloat[Mirrors of $D_{3v}$]{\includegraphics[height=2cm]{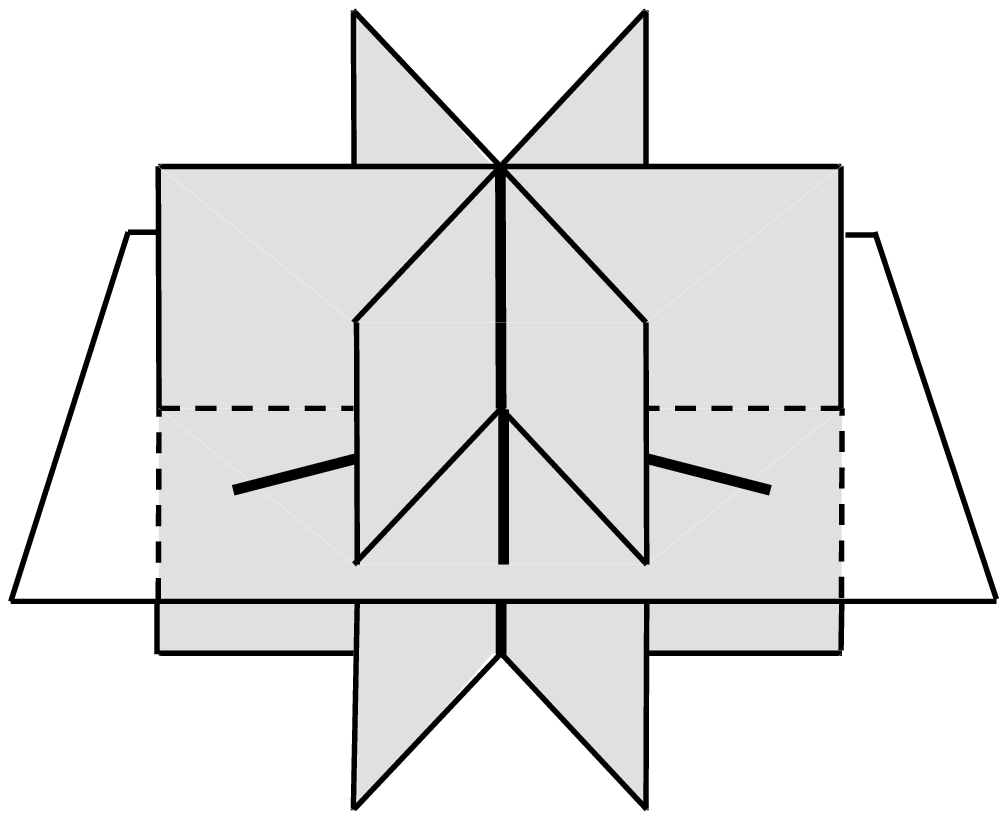}\label{fig:Dv-mirrors}}
 \quad 
\subfloat[Three mirrors of $O_h$]{\includegraphics[height=2cm]{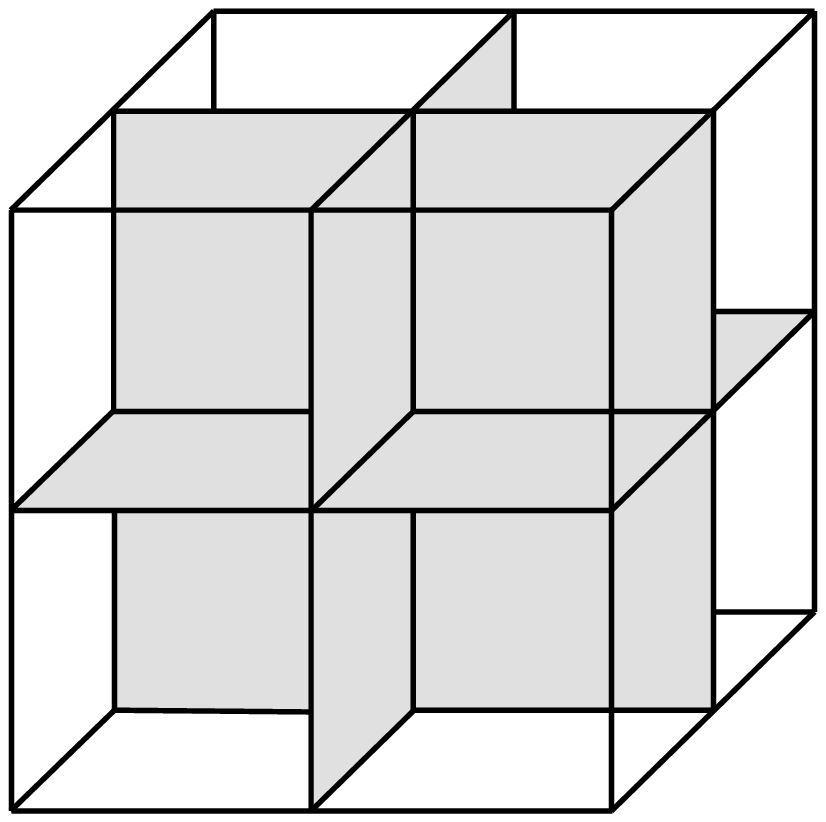}\label{fig:O-3mirrors}}
 \quad 
 \subfloat[One of the six mirrors of $O_h$]{\includegraphics[height=2cm]{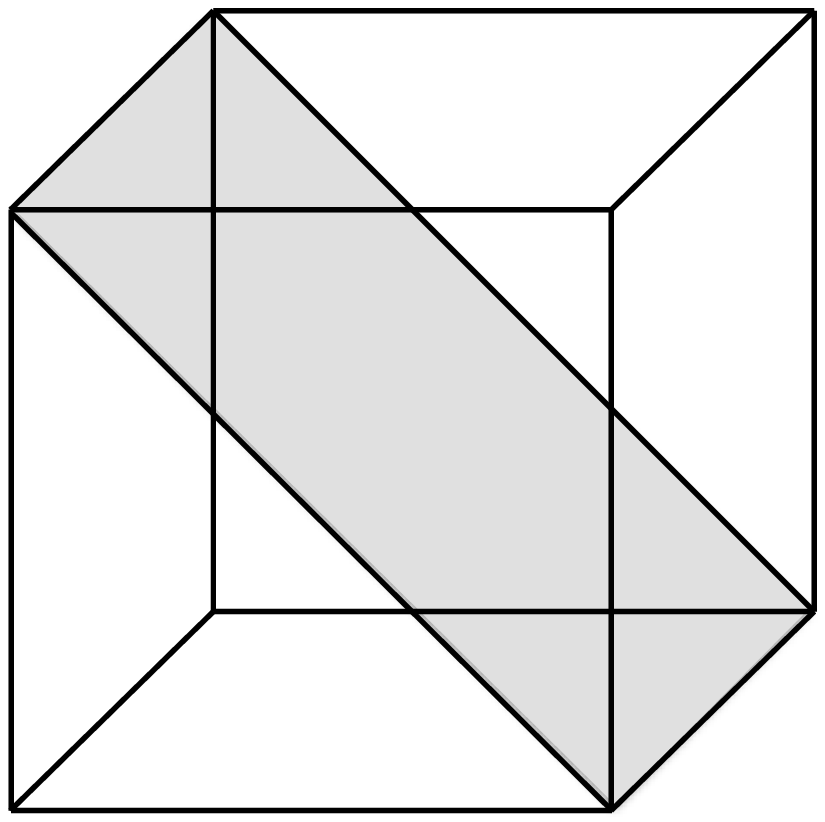}\label{fig:O-6mirrors}}
 \caption{Mirrors in composite symmetry.
 The rotation axes (bold lines)
 and mirror planes of $(a)$ $D_{3h}$ and $(b)$ $D_{3v}$. 
 Since $T_d, T_h \prec O_h$,
 the mirror planes shows the arrangement of rotation axes and
 mirror planes of $(c)$ $T_h$ and $(d)$ $T_d$. 
 For example, the three mirror planes of $(c)$
 and the $3$-fold axes connecting the opposite vertices of a cube
 also show the arrangement of mirror planes and rotation axis of $T_h$. 
}
\label{fig:O-mirrors}
\end{figure}

\end{document}